\documentclass[12pt]{amsart}
\usepackage{geometry} 
\usepackage{graphicx}
\usepackage{amsfonts}
\usepackage{amsthm}
\usepackage{amsmath}
\usepackage{amssymb}
\usepackage[arrow,matrix,curve]{xy}
\usepackage{enumerate}
\usepackage{color}
\usepackage{multirow}
\usepackage{array}
\usepackage[normalem]{ulem}
\usepackage{tikz}
\newtheorem{theorem}{Theorem}
\newtheorem{corollary}{Corollary}

\newtheorem{proposition}{Proposition}
\theoremstyle{definition}

\newtheorem{definition}{Definition}

\theoremstyle{remark}
\newtheorem{remark}{Remark}
\newcommand{\co}{\colon\,}
\newcommand{\cA}{\mathcal A}
\newcommand{\cB}{\mathcal B}
\newcommand{\cU}{\mathcal U}
\newcommand{\bT}{\mathbb T}
\newcommand{\bR}{\mathbb R}
\newcommand{\bC}{\mathbb C}
\newcommand{\bH}{\mathbb H}
\newcommand{\bZ}{\mathbb Z}
\newcommand{\bP}{\mathbb P}

\newcommand{\cH}{\mathcal H}
\newcommand{\cK}{\mathcal K}

\newcommand{\pt}{\text{pt}}
\newcommand{\lp}{\textup{(}}
\newcommand{\rp}{\textup{)}}

\newcommand{\tKR}{\widetilde{KR}}

\begin{document}
\title[Orientifolds and twisted $KR$-theory]
{T-duality for orientifolds and twisted KR-theory}

\author{Charles Doran \and Stefan M\'{e}ndez-Diez}
\address{Department of Mathematical and Statistical Sciences\\
University of Alberta\\
Edmonton, AB T6G 2G1, Canada}
\email[Charles Doran]{doran@math.ualberta.ca}
\email[Stefan Mendez-Diez]{sdiez@math.ualberta.ca}
  
\author{Jonathan Rosenberg}
\address{Department of Mathematics\\
University of Maryland\\
College Park, MD 20742-4015, USA} 
\email[Jonathan Rosenberg]{jmr@math.umd.edu}
\thanks{Partially supported by NSF grant DMS-1206159.}
\keywords{orientifold, $O$-plane, $KR$-theory, T-duality, Chan-Paton
  bundle, brane charge} 
\subjclass[2010]{Primary 19L50; Secondary 19L47 81T30 19L64 19M05} 
\begin{abstract}
D-brane charges in orientifold string theories are classified by the
$KR$-theory of Atiyah. However, this is assuming that all $O$-planes
have the same sign. When there are $O$-planes of different signs,
physics demands a ``$KR$-theory with a sign choice'' which up until
now has not been studied by mathematicians (with the unique exception 
of Moutuou, who didn't have a specific application in mind). We give a
definition of 
this theory and compute it for orientifold theories compactified on
$S^1$ and $T^2$. We also explain how and why additional ``twisting'' is
implemented. We show that our results satisfy all possible T-duality
relationships for orientifold string theories on elliptic curves,
which will be studied further in subsequent work. 
\end{abstract}

\maketitle

\section{Introduction}
\label{sec:intro}

The purpose of this paper is to describe the versions of $K$-theory
needed to describe $T$-duality for orientifolds, and to compute and
analyze them in a few simple but important cases. By \emph{orientifolds}
we mean spacetimes of the form $\bR^k\times X$, where $X$ is a smooth
$10-k$ dimensional (oriented) manifold equipped with an involution, $\iota$,
which defines the orientifold structure.\footnote{Note that in some of
  the literature, the word ``orientifold'' is used to denote the
  quotient space $X/\iota$, but 
  it is really essential to keep track of the \emph{pair} $(X,\iota)$
  and not just the quotient.} 
\emph{Orientifold string theories} are defined by sigma-models with target
space an orientifold $(X, \iota)$, where the fundamental strings are
\emph{equivariant} maps $\varphi\co \Sigma\to X$, so
that $\iota\circ\varphi = \varphi\circ\Omega$. Here $\Sigma$ is an
oriented $2$-manifold, possibly with boundary (the case of open strings),
called the \emph{string worldsheet}, and $\Omega$, called the
\emph{worldsheet parity operator}, is an orientation-reversing
involution on $\Sigma$. We require $\Sigma/\Omega$, though not
necessarily $\Sigma$ itself, to be connected.  (Thus an allowable
possibility is $\Sigma = \Sigma_0 \amalg \overline{\Sigma}_0$, where
$\Sigma_0$ is a connected oriented surface, $\overline{\Sigma}_0$ is the same
surface with orientation reversed, and $\Omega$ interchanges the two.)
See for example \cite{Distler:2009ri}; there some extra twisting data,
which we are ignoring for the moment, is also taken into account, and
the notation is slightly different.

As described in \cite{Witten:1998,Minasian:1997}, $D$-branes in string theories
are classified by $K$-theory, where the relevant type of $K$-theory
depends on the string theory being considered. Since the physics of
$T$-dual theories is indistinguishable, the groups classifying stable
$D$-branes in two $T$-dual theories must be isomorphic. This led
Bouwknegt, Evslin, and Mathai, and later
Bunke and Schick, to describe the $T$-duality between the type IIB
theory on a spacetime $X$ that is a circle bundle over base $Z$, with
$H$-flux $H$, and the type IIA theory on a dual circle bundle $\widetilde{X}$
over $Z$, with dual $H$-flux $\widetilde{H}$, as an isomorphism of twisted
$K$-theories: 
\begin{equation}
\label{eqn:BSTdual}
K^*(X,H)\cong K^{*+1}(\widetilde{X},\widetilde{H}).
\end{equation}
In the above equation
\begin{equation}
c_1(X)=\widetilde{\pi}_*(\widetilde{H}) \text{  and  } 
c_1(\widetilde{X})=\pi_*(H),
\end{equation}
where $c_1(X)\in H^2(X;\bZ)$ is the first Chern class and
$\pi_*\co H^k(X)\to H^{k-1}(Z)$ is the Gysin push-forward map which
in terms of de Rham cohomology is defined by integration along
the fiber \cite{MR2080959,Bunke:2005}. This was later generalized to
the case where $X$ is a $T^n$-bundle in \cite{Mathai:2005} and
\cite{Bunke:2006}.   

For orientifolds, $D$-brane charges are classified by $KR$-theory
\cite[\S5.2]{Witten:1998}, \cite{Hori:1999me}, \cite{Gukov:1999}, which we will review in section
\ref{sec:reviewKR}. One benefit of using $KR$-theory is that it can be
viewed as a sort of \textit{universal $K$-theory}. It is ``universal''
in the sense that the $K$-theories
$KU$ for the type II theories, $KO$ for the type I
theory, and $KSC$ for the type $\widetilde{\text{I}}$ theory can
all be built out of $KR$. This shows that by keeping track of the
appropriate involution $\iota$, one does not need to make a choice of
which type of $K$-theory to use, and it is already accounted for just
by using $KR$-theory. However, $KR$-theory has some immediate limitations
that prevent us from generalizing a topological description of
$T$-duality like equation \eqref{eqn:BSTdual} to orientifolds.  

The first problem is that it is not immediately clear how to twist
$KR$-theory, or even what is meant by $H$-flux. Traditionally, $H\in
H^3(X;\bZ)$, so there is no reason to expect $H$ to be
equivariant. Another related issue is that orientifold theories
involve extra information, which is not just topological, relevant to
the stable $O$-plane charges. This issue is already apparent
when studying circle 
orientifolds, even though the dimension of a circle is too low to have
to worry about more general twistings.  

In section \ref{sec:orient} we will review $T$-duality between all
possible circle orientifolds and the classification of stable
$D$-branes in the different theories. The $T$-dual of the type I
theory on a circle (which is a type IIB orientifold on the circle with
trivial involution) is a type IIA orientifold on a circle with
involution given by reflection (referred to as the type IA or type I$'$
theory). The $T$-dual to the type IIB theory on the circle with the
antipodal map (sometimes referred to as the type $\widetilde{I}$
theory) is also 
a type IIA orientifold on the circle with involution given by 
reflection (often called the type $\widetilde{IA}$ theory). The
compactification manifolds for both the type IA and $\widetilde{IA}$
theories are topologically equivalent, with the difference being the
charges of the $O$-planes at the two fixed points. There are
physical descriptions of the classification of $D$-branes in the two
theories \cite{Olsen:1999,Bergman:1999}; however, we are not aware of
any mathematical description for the classification of $D$-branes in
the type $\widetilde{IA}$ theory via $KR$-theory. In fact, a
topological invariant such as $KR$-theory cannot pick up the
difference between the type IA and $\widetilde{IA}$ compactifications
since the distinction is non-topological. In section
\ref{sec:chargedKR} we propose a variant of $KR$-theory, which we
call $KR$-theory with a sign choice, that can distinguish between the
two cases, giving a mathematical description of the brane charges in
the type $\widetilde{IA}$ theory. We then give all possible sign
choices for $KR$-theories for orientifolds of $2$-tori.

A word about our sign convention: we say that an $O$-plane has
positive sign, or is an $O^+$-plane, if the Chan-Paton bundle on it
has orthogonal type, and has negative sign, or is an $O^-$-plane, if
the Chan-Paton bundle on it has symplectic type. The sign decorations
that we attach to $KR$-theory follow the same convention. Since a
tensor product of an orthogonal bundle with a symplectic bundle is
symplectic, while the tensor product of two symplectic bundles is
orthogonal, signs multiply as one would expect. This convention is the
same as the one made by Witten in \cite{Witten:1998-02}, but is the
\emph{reverse} of the convention made by Gao and Hori in
\cite{Gao:2010ava}. Both sign conventions are in general use, but we
feel that the multiplication rule indicates that this one is
preferable, even though it means (as Witten points out) that the
tadpoles are of opposite sign.

When we move up in dimension to $2$-tori, the sign choice is no longer
enough to account for all possible orientifold theories. In particular,
$KR$-theory with a sign choice cannot describe the type I theory
without vector structure \cite{Witten:1998-02}. For this we need to
include more general twists of $KR$-theory, which will be discussed in
section \ref{sec:generaltwists}. Here we use physics to
motivate which $KR$-theories should be isomorphic, and 
check the results via topology.
The twist applied to
$KR$-theory is related to the geometry of its $T$-dual theory and is
described in \cite{DMDR}. The purpose of the current paper is to
describe the relevant twisted $KR$-theories needed to give the
geometric interpretation in \cite{DMDR}. 

One of our motivations for a detailed analysis of $T$-duality via
orientifold plane charges in $KR$-theory was the special case of $c =
3$ Gepner models as studied in \cite{Bates:2006}.  The authors of that
paper used simple current techniques in CFT to construct the charges
and tensions of Calabi-Yau orientifold planes, though a $K$-theoretic
interpretation was missing.  Although the interpretation of brane
charges in $KR$-theory is sensitive to regions of stability, this
$K$-theoretic interpretation does not depend on the specific structure
of $c = 3$ Gepner models, nor even on a rational conformal field
theoretic description.  These results should be contrasted with the
recent work \cite{2010arXiv1012.1634E} where a twisted equivariant
$K$-theory description of the $D$-brane charge content for WZW models
is provided (see also \cite{MR2262682} for examples which make
explicit the isomorphism with topological $K$-theory in the case of
some Gepner models).  Work in progress seeks to establish an
isomorphism between a suitable (real) variant of twisted equivariant
$K$-theory, sufficient to capture orientifold charge content, and our
$KR$-theory with sign choices for Gepner models.  As a side-effect,
such an isomorphism will then permit computation of $KR$-theory for
complicated Calabi-Yau manifolds through a simpler computation at the
Gepner point. 

After the first version of this paper was completed, we became aware
of the work of Moutuou \cite{MoutuouThesis,2011arXiv1110.6836M,MR3158706}
on groupoid twisted $K$-theory, which includes
our $KR$-theory with a sign choice as a special case.  Indeed,
Moutuou's classification of possible twists of $KR$ coincides with ours,
though his point of view and motivation were quite different.

We would like to thank Max Karoubi for many useful discussions
regarding the contents of
Section \ref{sec:chargedKR}, and in particular for suggesting the
formulation of Theorem \ref{thm:KRspecies1}, as well as a method of
proof for that theorem. We also thank the referee for several useful
suggestions. 

\section{Review of classical $KR$-theory}
\label{sec:reviewKR}

Let $X$ be a locally compact space (in most physical situations it
will be a smooth manifold) with involution $\iota$. A Real vector
bundle on $X$ (in the sense of Atiyah \cite{MR0206940}) 
is a complex vector bundle $p\co E\to X$ together with
a conjugate-linear vector bundle isomorphism $\varphi\co E\to E$ such
that $\varphi^2=1$ and $\varphi$ is compatible with $\iota$, in the
sense that $p\circ \varphi = \iota\circ p$. $KR(X)$ is the group of
pairs of Real vector bundles $(E,F)$ on $X$
(with compact support) modulo the equivalence relation 
\begin{equation}
\label{eqn:equivrel}
(E,F)\sim (E\oplus H, F\oplus H),
\end{equation}
for any  Real vector bundle $H$. Note that
$KR(X)$ depends on the involution $\iota$ even though it isn't
explicitly stated. The compact support condition means that we can
choose $E$ and $F$ to be trivialized off a sufficiently large
$\iota$-invariant compact set, with $\varphi$ off this compact set
being standard complex conjugation on a trivial bundle. 

To define the higher $KR$-groups, $KR^{-j}(X)$, we must first
introduce some notation. Let $\bR^{p,q}=\bR^p+i\bR^q$, where the
involution is given by complex 
conjugation, and $S^{p,q}$ be the $p+q-1$ sphere in
$\bR^{p,q}$. \emph{Caution}: In
this notation, the roles of $p$ and $q$ are the reverse of those in the
notation used by Atiyah in \cite{MR0206940} but the same as the
notation in \cite{MR1031992}, \cite{Hori:1999me}, \cite{Bergman:1999} and
\cite{Olsen:1999}. Then we can define 
$$KR^{p,q}(X)=KR(X\times\bR^{p,q}).$$
This obeys the periodicity condition
$$KR^{p,q}(X)\cong KR^{p+1,q+1}(X),$$
so $KR^{p,q}$ only depends on the difference $p-q$ and we can define
$$KR^{q-p}(X)=KR^{p,q}(X).$$
$KR^{-j}(X)$ is periodic with period $8$.

When $\iota$ is the trivial involution, the Reality condition is
equivalent to $E$ being the complexification of a 
real bundle. Thus $KR$ gives a classification of real
vector bundles and we find 
\begin{equation}
\label{eqn:KRKO}
KR^{-j}(X)\cong KO^{-j}(X),
\end{equation}
when $\iota$ is trivial \cite[p.\ 371]{MR0206940}. Complex $K$-theory can also
be obtained from $KR$-theory using 
\begin{equation}
KR^{-j}(X\times S^{0,1}) = KR^{-j}(X\amalg X)\cong K^{-j}(X),
\label{eqn:KRK}
\end{equation}
where the involution exchanges the $2$ copies of $X$ \cite[Proposition
  3.3]{MR0206940}. And as shown by Atiyah \cite[Proposition
  3.5]{MR0206940},
$KR^{-j}(X\times S^{0,2})\cong KSC^{-j}(X)$, the self-conjugate $K$-theory
of Anderson \cite{Anderson} and Green
\cite{MR0164347}, which is periodic with period $4$.

In fact, when the involution $\iota$ has no fixed points, there is a
spectral sequence (\eqref{eqn:Bredon} below), whose $E_2$-term is
$4$-periodic, converging to $KR^{-j}(X)$. This motivated Karoubi and Weibel
\cite[Proposition 1.8]{Karoubi:2005} to assert that $KR^{-j}(X)$ is
always $4$-periodic when the involution is free, but in general this
is not the case (unless one inverts the prime $2$). The groups
$KR^{-j}(S^{0,4})$ provide a counterexample.

When $X$ is compact and $X^\iota$ is non-empty, the inclusion of an
$\iota$-fixed basepoint into $X$ is equivariantly split, so
the reduced $KR$-groups, $\widetilde{KR}^{-j}(X)$, are defined such that 
$$KR^{-j}(X)\cong\widetilde{KR}^{-j}(X)\oplus KR^{-j}(\pt).$$
We will write simply $KR^{-j}$ or $KO^{-j}$ for $KR^{-j}(\pt)$. When
$Y\subseteq X$ is closed and $\iota$-invariant, we can define the
relative $KR$-theory as 
$$KR^{-j}(X,Y)\cong\widetilde{KR}^{-j}(X/Y).$$
As we will discuss in the following section, this is the relevant
group for classifying $D$-brane charges. 

\section{Orientifolds on a circle and $T$-duality}
\label{sec:orient}

In this section we will consider orientifold string theories 
with target space $(S^1, \iota)$, where we view $S^1$ as the unit
circle in $\bR^2$ and where the involution $\iota$ comes from a linear
involution on  $\bR^2$. Since linear involutions are classified by the
dimension of the $(-1)$-eigenspace, there are (up to isomorphism)
exactly three possibilities for $(S^1, \iota)$: the trivial involution
corresponding to $S^{2,0}$, reflection corresponding to $S^{1,1}$, and
the antipodal map corresponding to $S^{0,2}$. $S^{2,0}$ and $S^{0,2}$
only support the type IIB theory since the involution is orientation
preserving, while $S^{1,1}$ only supports the type IIA theory since
the involution is orientation reversing.  

The type IIB theory on $S^{2,0}$ is the type I theory compactified on
a circle. It is known to be $T$-dual to the type IIA theory on
$S^{1,1}$, sometimes referred to as the type IA or I$'$ theory
\cite{Olsen:1999,Bergman:1999,Hori:1999me}. The type IIB theory on $S^{0,2}$ is
often called the type $\widetilde{\text{I}}$ theory and is $T$-dual to the type
$\widetilde{IA}$ theory \cite[\S6.2]{Witten:1998-02},
\cite[\S7.1]{Gao:2010ava}.   

In this section we will review these $T$-duality relations and their
$K$-theoretic descriptions. The lack of a mathematical description for
the $K$-theory description of the type $\widetilde{IA}$ theory will motivate the
definition for a variant of $KR$-theory given in Section
\ref{sec:chargedKR}.  Before describing the various $T$-dualities we
will review how $D$-branes are classified by $K$-theory. 

\subsection{Classification of $D$-branes by $KR$-theory}

$D$-branes on orientifolds $(X,\iota)$,
where $X$ is a smooth manifold and $\iota$ is an involution on $X$, 
are classified by pairs of vector bundles on $X$ (the Chan-Paton
bundles), each with conjugate-linear involutions compatible with $\iota$, 
modulo creation and annihilation of charge zero
$D$-brane systems (as in equation \eqref{eqn:equivrel}). So
$D$-branes in orientifolds are classified by $KR$-theory
\cite[\S5.2]{Witten:1998}. 

More generally, when we compactify string theory on an $m$-dimensional
space $M$, so that the spacetime manifold is $\bR^{10-m,0}\times M$,
we are interested in the charges of $D$-branes in the non-compact
dimensions. So we want to consider $D$-branes of codimension $9-m-p$
in $\bR^{9-m,0}$. These can arise from both $Dp$-branes located at a
particular point in $M$ or higher dimensional $D$-branes that wrap
non-trivial cycles in $M$. Furthermore, we only want to consider
branes with finite energy, so we want to classify bundle pairs that are
asymptotically equivalent to the vacuum in the transverse space
$\bR^{9-m-p,0}$. Mathematically this means
we want to add a copy of $M$ at infinity, i.e., take (the one-point
compactification of $\bR^{9-m-p,0}$)$\times M$, and consider bundles on
$S^{10-m-p,0}\times M$ that are trivialized on the copy of $M$ at
infinity. Such bundles are classified by $KR^{-i}(S^{10-m-p,0}\times
M, M)$; the index $i$ depends on the string theory and involution being
considered. For purposes of calculation it is useful to relate this
to the $KR$-theory of $M$.  
\begin{proposition}
\label{Thm:relvabs}
$$KR^{-i}(S^{10-m-p,0}\times M, M)\cong KR^{p+m-9-i}(M).$$
\end{proposition}
\begin{proof}
Note that by excision,
\begin{align*}
KR^{-i}(S^{10-m-p,0}\times M,M) &\cong
KR^{-i}\left(\left(S^{10-m-p,0}\smallsetminus \{\pt\}\right) \times
  M\right)\\
&\cong KR^{-i}\left(\bR^{9-m-p,0} \times M\right)\\
&\cong KR^{p+m-9-i}(M). \qedhere
\end{align*}
\end{proof}

Thus $Dp$-branes are classified by $KR^{p+m-9-i}(M)$. It is important to
keep track of the index. This point is often overlooked when studying
$D$-branes in the (non-orientifold) type II theories, which are classified by
$KU$-theory, since $KU$-theory has period $2$ and only the parity of
the index matters.

In what follows we will also have to study the charges of the
\emph{$O$-planes}, the components of the fixed set of the involution
on spacetime.\footnote{This terminology is unfortunate but standard;
  $O$-planes in general orientifold theories
  do not have to be planes. They can have more complicated
  topology.} The restriction of a Chan-Paton bundle to an $O$-plane
must have either a real (positive) or symplectic (negative) structure.
The classification of $D$-branes via $KR$-theory is only valid when all
$O$-planes have positive charge, and breaks down when different
$O$-planes have different charges. It is this
breakdown that leads us to define $KR$-theory with a sign choice in section
\ref{sec:chargedKR}. 

\subsection{The Type I Theory and Its $T$-dual}
\label{sec:typeITdual}

The type I theory compactified on a circle is formally identical
to the type IIB orientifold theory compactified on $S^{2,0}$. Consider the
bosonic fields in the type IIB theory compactified on $S^{2,0}$, 
$$X=X_L+X_R.$$
The worldsheet parity operator reverses the orientation of the string,
and so exchanges left-movers and right-movers. This leaves the bosonic
fields invariant under $\Omega$ and therefore compatible with the
trivial involution. 

$T$-duality leaves the left-moving fields
invariant, while reversing the sign of the right-moving fields, so the
$T$-dual coordinates are 
$$\tilde X=X_L-X_R.$$
Under the action of $\Omega$, the $T$-dual coordinates transform as
$$\tilde X\mapsto -\tilde X.$$
This shows that the $T$-dual to the type I theory compactified on a
circle must be the type IIA theory (since $T$-duality exchanges types
IIA and IIB) mod the action of $\Omega$ combined with the spacetime
involution that reflects the compact dimension. This is the type IIA
theory compactified on $S^{1,1}$. In the literature it is often
referred to as the type IA (or I$'$) theory. One could also show these
two theories are $T$-dual to one another by showing there is no
momentum, but winding in the $S^{2,0}$ direction, while $S^{1,1}$ has
momentum, but no winding. 

The type I theory on $S^1$ has a space filling $O9^+$-plane wrapping
the compact dimension. The $T$-dual type IA theory has
$2$ $O8^+$-planes located at the $2$ fixed points of $S^{1,1}$. Recall that we
use the plus sign to denote that the $O$-planes have negative
$D$-brane charge and require the addition of $D$-branes to obtain a
zero charge system.  

$Dp$-brane charges in the type I theory compactified on a circle are
classified by 
\begin{align}
\label{eqn:KRbranes_I}
KR(S^{9-p,0}\times S^{2,0}, S^{2,0}) &\cong KO^{p-8}(S^1)\nonumber\\
&\cong KO^{p-8}\oplus KO^{p-9}.
\end{align}
The second factor in the last line of equation \ref{eqn:KRbranes_I}
corresponds to $Dp$-brane charge coming from unwrapped branes and the
first factor corresponds to the charge contribution from branes
wrapping $S^1$. The complete brane content is given in Table 
\ref{Table:species2_S1}. 

Since the type IA theory is obtained from the type I theory
compactified on a circle by a $T$-duality, the relevant $KR$-theory is
shifted in index by $1$. Therefore, $Dp$-brane charges in the type IA
theory are classified by 
\begin{align}
KR^{-1}(S^{9-p,0}\times S^{1,1}, S^{1,1}) &\cong KR^{p-9}(S^{1,1})\nonumber\\
&\cong KO^{p-9}\oplus KO^{p-8},
\end{align}
where the second factor on the right-hand side corresponds to
$Dp$-brane charge coming from unwrapped branes and the first factor
corresponds to the charge contribution from wrapped branes. The
complete brane content is given in Table \ref{Table:species2_S1}. The
fact that $T$-duality exchanges wrapped and unwrapped branes is
described by the exchanged roles for $KO^{p-8}$ and $KO^{p-9}$ in the two
theories. 

\begin{table}[ht]
\noindent\makebox[\textwidth]{\begin{tabular}{|| c || c | c | c | c | c | c | c | c | c | c || m{2.25cm} | m{2.25cm} ||}
\hline
$Dp$-brane & $D8$ & $D7$ & $D6$ & $D5$ & $D4$ & $D3$ & $D2$ & $D1$ & $D0$ & $D(-1)$ & type I on $S^1$ & type IIA on $S^{1,1}$ \\
\hline\hline
$KO^{p-8}$ & $\bZ$ & $\bZ_2$ & $\bZ_2$ & $0$ & $\bZ$ & $0$ & $0$ & $0$ & $\bZ$ & $\bZ_2$ & $(p+1)$-brane wrapping $S^{2,0}$  & unwrapped $p$-brane\\ \hline
$KO^{p-9}$ & $\bZ_2$ & $\bZ_2$ & $0$ & $\bZ$ & $0$ & $0$ & $0$ & $\bZ$ & $\bZ_2$ & $\bZ_2$ & unwrapped $p$-brane & $(p+1)$-brane wrapping $S^{1,1}$  \\ \hline
\end{tabular}}
\smallskip
\caption{$D$-brane charges in the type I theory compactified on a circle and the type IA theory.}
\label{Table:species2_S1}
\end{table}

The non-BPS torsion charged branes are not stable at all points of the moduli
space. $D0$-brane charge in the type I theory receives an integral
contribution from a wrapped BPS $D1$-brane and a $\bZ_2$ contribution
from an unwrapped non-BPS $D0$-brane. 
$K$-theory accurately predicts the entire brane charge spectrum 
everywhere, in and out of the region of stability for the
non-BPS branes, but the sources of the charges may vary at different
points of the moduli space.  For more details, see \cite{Bergman:1999}.

\subsection{The Type $\widetilde{\text{I}}$ and $\widetilde{IA}$ Theories}

The type $\widetilde{\text{I}}$ theory is the type IIB orientifold
$(\bR^9\times S^1, \iota)$ where $\iota$ is the spacetime
involution that rotates $S^1$ by $\pi$ radians. In our
notation, this is the type IIB theory on $\bR^{9,0}\times S^{0,2}$. The
$T$-dual of the type $\widetilde{\text{I}}$ theory is the type $\widetilde{IA}$
theory \cite{Bergman:1999}. As we saw in the last section, the type IA
theory contains $2$ 
$O8^+$-planes. The type $\widetilde{IA}$ theory is obtained from the
type IA theory by replacing one of the $O8^+$-planes with an
$O8^-$-plane. Here an $O^-$-plane is an $O$-plane with symplectic
Chan-Paton bundle and positive
$D$-brane charge. (Note that if there were $O8^-$-planes at both fixed
points, then a charge $0$ system would require the addition of
anti-branes and wouldn't be supersymmetric.) 
We will refer to the compactification circle as
$S_{(+,-)}^{1,1}$. It is topologically equivalent to a compactification on
$S^{1,1}$, in that there are $2$ fixed points. However, the net
$O$-plane charge is zero.

$Dp$-brane charges in the type $\widetilde{\text{I}}$ theory are classified by 
\begin{equation}
KR(S^{9-p,0}\times S^{0,2},S^{0,2})\cong KSC^{p-8}.
\end{equation}
$KSC$ doesn't split into pieces from wrapped and unwrapped
branes as in the 
previous case. The authors of \cite{Bergman:1999} were still able
to determine which 
charges come from wrapped and unwrapped branes using what we know
about $T$-duality, the type IA theory and $O8^\pm$-planes. 

Since the type $\widetilde{IA}$ theory is $T$-dual to the type
$\widetilde{\text{I}}$ 
theory, $Dp$-brane charges in the type $\widetilde{IA}$ theory must
also be classified by $KSC^{p-8}$. It is important to note that there
is no mathematical description for this that we are aware of. There is
only the physical reasoning, which requires the assumption of
$T$-duality. Since the underlying topological space
for the type $\widetilde{IA}$ theory is $S^{1,1}$, we should be able
to classify $D$-brane charges by some \textit{twisted} $KR$-theory of
$S^{1,1}$. This idea motivates the definition given in the following
section.

\section{$KR$-theory with a Sign Choice}
\label{sec:chargedKR}

The compactification manifolds for the type IA and $\widetilde{IA}$
theories are topologically equivalent, even taking the involution
$\iota$ into account. Therefore, $KR$-theory cannot
differentiate between them. These two physical theories are
differentiated by the signs of the $O$-planes located at their fixed
sets, so we must enhance $KR$-theory with this information. 

Along with the space $X$ and the action of a group $G$ (in our case
$\bZ_2$), we must also include a sign choice, $\alpha$, on the
components of the fixed set. Physically this sign choice determines
the type of $O$-plane at the different components of the fixed set. In
other words, it is a choice of orthogonal or symplectic Chan-Paton
bundles on the different components. Recall our convention that a $+$
choice corresponds to an orthogonal Chan-Paton bundle, and a $-$
choice to a symplectic one. Note that the fixed sets for the
type IA and $\widetilde{IA}$ theories both have $2$ components, each a
point. The 
type IA theory is the sign choice $\alpha=(+,+)$, while the type
$\widetilde{IA}$ theory is the sign choice $\alpha=(+,-)$. We define
an extension of $KR$-theory that contains this information and that
fits into an exact sequence as in Theorem \ref{thm:Krsignles} below.

Intuitively, $KR_\alpha$ theory is defined in terms of a
generalization of Real vector bundles, namely pairs $(E,\Phi)$, where
$E$ is a complex vector bundle
over a real space $(X,\iota)$, and $\Phi\co E\to E$ is a conjugate-linear
vector bundle automorphism, equivariant with respect to $\iota$, 
and with $\Phi^2$ given by multiplication by 
$+1$ on components of the fixed set with a $+$ sign,
$-1$ on components of the fixed set with a $-$ sign.  Of course, if
all components of the fixed set have a $+$ sign and $\Phi^2\equiv 1$, then
this is just Atiyah's definition of a Real vector bundle. If all
components of the fixed set have a $-$ sign and $\Phi^2\equiv -1$, then 
this is the corresponding notion in the symplectic case (used to
define the theory often called $KH$ --- this is the name introduced
in \cite{Gukov:1999}, but the theory already appeared much earlier
in \cite{MR0254839}).  But for sign choices with both
signs present, it is not clear how changing $\Phi^2$ changes the
notion of Real$_\alpha$ vector bundle, or how
to get from this rough definition to a theory satisfying
Bott periodicity.  So for all these reasons
(as in \cite{MR1827946} and the literature on
twisted $K$-theory, for example), a rigorous definition of $KR_\alpha$
requires noncommutative geometry.

Therefore what we really do is
to define $KR_\alpha(X)$ to be the topological $K$-theory of a certain
noncommutative Banach algebra $\cA_\alpha(X)$. 
In what follows, $\cK$ and $\cK_{\bR}$
denote the algebras of compact operators on an infinite-dimensional separable
complex Hilbert space and an infinite-dimensional separable real
Hilbert space, respectively. Before we get to the rigorous definition
of $KR_\alpha(X)$, we first show that the requisite Banach algebras
exist, and study their topological $K$-theory.

\begin{theorem}
\label{thm:KRsignexists}
Let $(X,\iota)$ be a Real locally compact space with an assignment
of signs $\alpha$ to the components of the fixed set. Then an algebra
$A_\alpha(X)$ exists satisfying the following properties:
\begin{enumerate}
\item $A_\alpha(X)$ is a {\bfseries real} continuous-trace $C^*$-algebra
whose complexification \linebreak
$A_\alpha(X)\otimes_{\bR} \bC$ has spectrum $X$ and trivial
Dixmier-Douady invariant, and for which the induced action $\sigma$ of
$\text{Gal}\, (\bC/\bR)$ on $X$ is the given involution on $X$.  
\item The quotient of $A_\alpha(X)$ associated to any
component $Y^+$ of 
$X^G$ with positive sign choice is Morita equivalent {\lp}over
$Y^+${\rp} to $C_0^\bR(Y^+)$.
\item The quotient of $A_\alpha(X)$ 
associated to any component $Y^-$ of 
$X^G$ with negative sign choice is Morita equivalent {\lp}over
$Y^-${\rp} to $C_0^\bH(Y^-)$, where $\bH$ denotes the quaternions.
\end{enumerate}
Furthermore, there is a {\bfseries canonical} choice $\cA_\alpha(X)$
of such an algebra $A_\alpha(X)$.
\end{theorem}

\begin{proof}
Let $Y^+$ be the union of
the components of the fixed set with $+$ sign choice, $Y^-$ be the
union of the components of the fixed set with $-$ sign choice, and $Z
= X \smallsetminus( Y^+ \amalg Y^-)$, which is the open subset of $X$ 
on which the involution $\iota$ acts freely. Let $\cA(Z)$
denote the commutative real $C^*$-algebra $\cA(Z) = \left\{ f\in C_0(Z)
\mid f(\iota (x))= \overline{f(x)} \right\}$. {\lp}Recall that the
$K$-theory of $\cA(Z)$ is identical to $KR^*(Z)$.{\rp}.
First we will show that there is a spectrum-fixing isomorphism of
real $C^*$-algebras 
\[
\varphi\co \cA(Z) \otimes_{\bR} \cK_{\bR} \xrightarrow{\cong} \cA(Z) 
 \otimes_{\bR} \bH \otimes_{\bR} \cK_{\bR}.
\]
{\lp}The induced isomorphism on complexifications is equivariant for the
involution $\sigma$.{\rp} 
In fact, there is a canonical choice for $\varphi$ (up to homotopy).
Then we can define $A_\alpha(X)$ by ``clutching.'' 

The algebra $\cA(Z) \otimes_{\bR} \cK_{\bR}$ is, as
explained in \cite[\S3]{MR1018964}, the algebra of sections
(vanishing at infinity on $Z$) of a bundle over $\overline Z =
Z/\iota$ of real $C^*$-algebras with fibers
$\cK$ and structure group $P\cU'$, the projective infinite-dimensional
unitary/antiunitary group. This group is a semidirect product of
$P\cU$ by $\bZ_2$ (acting by complex conjugation), and the bundle is
induced from the $\bZ_2$-bundle 
$Z\to \overline Z$ defined by the free involution $\iota$. Now
$\bC\otimes_{\bR}\bH \cong M_2(\bC)$, so $\cA(Z) 
\otimes_{\bR} \bH \otimes_{\bR} \cK_{\bR}$ is also the algebra of
sections of a bundle over $\overline Z$ with fibers $M_2(\bC)\otimes
\cK\cong \cK$ and the same structure group, and since the bundle came
from the original $P\cU'$-bundle (via tensoring with $\bH$) and
induces the same covering map $Z\to \overline Z$, the bundles are
isomorphic (as $P\cU'$-bundles). This guarantees existence of the
desired isomorphism $\varphi$. In fact, if we fix an isomorphism
$\bC\otimes_{\bR}\bH \otimes \cK \to \cK$ (which is unique up to
homotopy), we get a canonical choice of $\varphi$, also unique up to
homotopy. 

Now, let $X^+=Z\cup
Y^+$, $X^-=Z\cup Y^-$, which are both open subsets of $X$. $\cA(Z)$ is
an ideal in each of the commutative real $C^*$-algebras 
$\cA(X^\pm) = \left\{ f\in C_0(X^\pm)  \mid 
f(\iota (x))= \overline{f(x)} \right\}$. We can construct $A_\alpha(X)$ as
the algebra of sections of a bundle of algebras obtained by clutching
the stabilized bundles for $\cA(X^+)$ and for $\cA(X^-)\otimes \bH$
together over $\overline Z$ via (the bundle isomorphism associated to)
$\varphi$, i.e., we construct $A_\alpha(X)$ by gluing
$\cA(X^+)\otimes _{\bR}\cK_{\bR}$ {\lp}which represents $KR^*(X^+)${\rp} to
$\cA(X^-)\otimes _{\bR}\bH \otimes _{\bR}\cK_{\bR}$ {\lp}which represents
$KR^*_\alpha(X^-)\cong KSp^*(X^-)${\rp} over $Z$ using $\varphi$. It
remains to show that we can choose $\varphi$ so that the
Dixmier-Douady invariant of $A_\alpha(X)\otimes_{\bR}\bC$
vanishes. This follows from the Mayer-Vietoris sequence for the
diagram
\[
\xymatrix{& X^+ \ar[rd] & \\ Z \ar[ur] \ar[dr] & & X\\
& X^- \ar[ur] & }
\]
since the Dixmier-Douady invariant is trivial over $X^+$ and $X^-$ (by
construction) and thus the Dixmier-Douady invariant in $H^3(X)$ comes
from the Phillips-Raeburn invariant of $\varphi$ in $H^2(Z)$ via the
Mayer-Vietoris boundary map. This invariant will be trivial for the
canonical choice.  (See the discussion in Remark
\ref{rem:nonunique} below for more details.)
\end{proof}

\begin{remark}
\label{rem:nonunique}
One has to be cautious; even though Theorem \ref{thm:KRsignexists}
guarantees \emph{existence} of $A_\alpha(X)$, it does not guarantee
\emph{uniqueness}, since the isomorphism $\varphi$ is only determined
up to an automorphism of the $P\cU'$-bundle over $\overline Z$. 
Such an automorphism, which (if $\overline Z$ is connected) we can
assume is in the connected component of the identity in the
automorphism group, is simply a section of the bundle of topological
groups $\cB_{P\cU} = (Z \times_{\overline Z} P\cU) \to \overline Z$, where the
covering group $\bZ_2$ acts on $\cU$ and thus on $P\cU$ by complex
conjugation. The automorphism will not affect the $K$-groups if it is
\emph{inner}, i.e., comes from a section of  $\cB_{\cU} = (Z \times_{\overline Z}
\cU) \to \overline Z$. From the exact sequence in sheaf cohomology for
the exact sequence of sheaves of groups
\[
1 \to \cB_{\bT} \to \cB_{\cU}  \to \cB_{P\cU} \to 1,
\]
where $\cB_{\bT}  = (Z \times_{\overline Z} \bT) \to \overline Z$ and
we identify bundles of topological groups with their sheaves of
sections, and from the fact that the sheaf $\cB_{\cU}$ is fine since
$\cU$ is contractible, we see that the obstruction to an automorphism
being inner lies in $H^1(\overline{Z}, \cB_{\bT}) = H^2(\overline{Z}, 
\uwave{\bT})$, where $\uwave{\bT}$ is the sheaf of local sections of
$\cB_{\bT}$. The obstruction group is via the exact sequence of sheaves 
\[
0 \to \uwave{\bZ} \to \uwave{\bR} \to \uwave{\bT} \to 1
\]
identifiable with $H^2(\overline{Z},\uwave{\bZ})$, where $\uwave{\bZ}$
is the locally constant sheaf with stalks $\bZ$ and twisting given by
the covering map $Z\to \overline{Z}$. The obstruction is what we can
call the \emph{twisted Phillips-Raeburn invariant} (cf.\
\cite[\S1]{MR1018964}). It vanishes when
$H^2(\overline{Z},\uwave{\bZ}) = 0$, and in particular when $\dim Z =
1$, so in this case $A_\alpha(X)$ is unique up to spectrum-fixing
Morita equivalence. $\square$
\end{remark}

While we will always use the particular algebra $\cA_\alpha(X)$
constructed in the proof of Theorem \ref{thm:KRsignexists}, any
other algebra satisfying the properties in the Theorem
gives the same $K$-groups up to extensions. 

\begin{theorem}
\label{thm:Krsignles}
For any $A_\alpha(X)$ with the properties of 
\textup{Theorem \ref{thm:KRsignexists}},
the topological $K$-groups fit into the long exact sequence 
\begin{equation}
\label{eqn:twKRles}
\cdots\to KR^{-i}(Z)\to K_i(A_\alpha(X))
  \to KO^{-i}(Y^+)\oplus KSp^{-i}(Y^-)\to KR^{-i+1}(Z)
  \to \cdots,
  \end{equation}
and the groups $K_*(A_\alpha(X))$ are uniquely determined
at least up to extensions.
\end{theorem}

\begin{proof}
Let $A_\alpha(Z)$ be the ideal of $A_\alpha(X)$ associated to $Z$, and
let $A_\alpha(Y^\pm)$ be the quotient associated to $Y^\pm$.
The long exact sequence  
\begin{equation}
\cdots\to K_i(A_\alpha(Z))\to
  K_i(A_\alpha(X))\to K_i(A_\alpha(X\smallsetminus Z))\to K_{i-1}(A_\alpha(Z))
  \to \cdots.
  \end{equation}
follows from the long
exact $K$-theory sequence of the extension of real $C^*$-algebras
associated to the open inclusion $Z\subseteq X$ (see for example
\cite[equation ($*$), p.\ 376]{MR1018964}). Since the involution
$\iota$ on $Z$ is free, $K_i(A_\alpha(Z))\cong KR^{-i}(Z)$.  Also 
\begin{align*}
K_i(A_\alpha(X\smallsetminus Z)) &\cong K_i(A_\alpha(Y^+\amalg Y^-))\\
&\cong K_i(A_\alpha(Y^+))\oplus K_i(A_\alpha((Y^-)).
\end{align*}
But
$$K_i(R_\alpha(Y^+))\cong KO^{-i}(Y^+),$$
since $Y^+$ has trivial involution. 

For a space $M$ where all the components of $M^G$ have $-$ sign
choice, the only difference is that the quotient of $A_\alpha(M)$  
associated to any component of $M^G$ is Morita equivalent (over
$\bR$) to $C_0^\bH(M^G)$ (instead of $C_0^\bR(M^G)$). This defines a
symplectic structure (instead of a real structure). In this case,
$KR$-theory with a sign choice reduces to what is sometimes referred
to as $KSp$- or $KH$-theory, which is just ordinary $KO$-theory with a
shift in index by $4$. Therefore,  
\[
K_i(A_\alpha(Y^-))\cong  KO^{-i-4}(Y^-) \cong KSp^{-i}(Y^-).
\]
Putting this all together gives the long exact sequence \eqref{eqn:twKRles}.

Since the connecting maps 
\[
KO^{-i}(Y^+) \to KR^{-i+1}(Z)\quad \text{\textup{and}}\quad
KSp^{-i}(Y^-) \to KR^{-i+1}(Z)
\]
in \eqref{eqn:twKRles} are determined
by the $KR$-theories of $X^+$ and $X^-$, respectively,
we conclude that regardless
of what choice one makes of $A_\alpha(X)$ satisfying the conditions
of Theorem \ref{thm:KRsignexists},
the groups $K_*(A_\alpha(X))$ are uniquely determined
at least up to extensions.
\end{proof}

\begin{definition}
\label{def:Krsign}
Let $X$ be a Real locally compact space with an assignment of signs
$\alpha$ to the components of the fixed set. Let $\cA_\alpha(X)$
be the canonical real continuous-trace algebra constructed in
\textup{Theorem \ref{thm:KRsignexists}}. We define $KR_\alpha^*(X)$
to be the topological $K$-theory of $\cA_\alpha(X)$ (in the sense of
\cite[\S3]{MR1018964}). Note that these groups fit into the exact
sequence given in \textup{Theorem \ref{thm:Krsignles}}.
\end{definition}

\begin{corollary}
$KR_\alpha^{-j}$ has periodicity with period $8$.
\end{corollary}

\begin{proof}
This is immediate from Bott periodicity for topological $K$-theory of
real Banach algebras.
\end{proof}

It is now easy to see that
$KR$-theory with a sign choice can be computed using a
generalization of the equivariant Atiyah-Hirzebruch spectral sequence 
of \cite[(A.2)]{Karoubi:2005}
\begin{equation}
\label{eqn:Bredon}
E_2^{p,q}=H^p_G(X;\uwave{KR}^q)\Rightarrow KR^{p+q}(X),
\end{equation}
where $\uwave{KR}^*$ is the Bredon coefficient system for $G$ associated to
$KR$. 

As described in \cite{Karoubi:2005}, $\uwave{KR}^*(G)=K^*$ and
$\uwave{KR}^*(\pt)=KO^*$. We are now allowing for different components
of the fixed set to have symplectic or  orthogonal bundles,
corresponding to the coefficient system being $KO$ or $KSp$.  
\begin{theorem}
\label{thm:chargedBredonSS}
There is a spectral sequence 
\begin{equation}
\label{eqn:chargedBredon}
E_2^{p,q}=H^p_G(X;\uwave{KR_\alpha}^q)\Rightarrow KR_\alpha^{p+q}(X).
\end{equation}
where $\uwave{KR_\alpha}^*(G)\cong K^*$ and
$$\uwave{KR_\alpha}^{-i}(\pt_j) = \left\{
        \begin{array}{ll}
          KO^{-i},   & \text{ if }\alpha_j=+, \\
           \strut KSp^{-i},   & \text{ if }\alpha_j=-.
          \end{array}\right.$$
\end{theorem}
\begin{proof}
The proof is quite similar to the case handled in
\cite{Karoubi:2005}. We filter $K^*_\alpha(X)$ using the equivariant
skeletal filtration, but with fixed cells separated into two types.
Then this is just the spectral sequence associated to this
filtration. The picture of the coefficient system is as in Figure
\ref{Fig:S11} (right side). 
\end{proof}

If we remove the $2$ fixed points from $S^{1,1}$ we are left with $2$
copies of $\bR$ that are exchanged by the involution. This gives
$KR^*_\alpha(S^{1,1}\smallsetminus \text{fixed points})\cong K^{*-1}$ by 
\eqref{eqn:KRK}. The type
IA theory has $O8^+$-planes (hence orthogonal bundles) at both fixed
points, so has $KO^*$ at both fixed points (see Figure \ref{Fig:S11})
and matches with the spectral sequence as described in
\cite{Karoubi:2005}. While motivated by physics, we are just
decorating $G$-$CW$-complexes with some extra information on the
equivariant cells of the form $(G/G) 
\times e^n$ that we have called ``sign.'' We show in \cite{DMDR} 
that the sign can be given a geometric interpretation in the $T$-dual
theory, thus giving a completely mathematical description of $T$-duality.

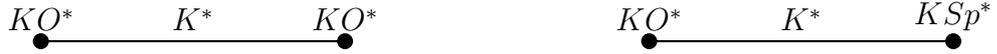
\begin{figure}
\begin{tikzpicture}[scale=4]
\draw [thick] (0,0) -- (1,0);
\fill (0,0) circle (0.75pt);
\fill (1,0) circle (0.75pt);
\node [above] at (.5,0) {$K^*$};
\node [above] at (0,0) {$KO^*$};
\node [above] at (1,0) {$KO^*$};
\draw [thick] (2,0) -- (3,0);
\fill (2,0) circle (0.75pt);
\fill (3,0) circle (0.75pt);
\node [above] at (2.5,0) {$K^*$};
\node [above] at (2,0) {$KO^*$};
\node [above] at (3,0) {$KSp^*$};
\end{tikzpicture}
\caption{Coefficient systems for the type IA and $\widetilde{IA}$ theories.}
\label{Fig:S11}
\end{figure}

Note that flipping the sign of every component of the fixed set 
exchanges $KO$ and $KSp$ and so just results in a shift of the
index by $4$. For example, type IIA theory on $S^{1,1}$ with $2$
$O8^-$-planes does not make physical sense since it is not
supersymmetric, and it is mathematically uninteresting since it is just
the usual theory with an index shift. 
 
We can now turn our attention to the only case of $KR$-theory of a
circle with a non-trivial sign choice, corresponding to the type
$\widetilde{IA}$ theory. 

\subsection{The Type $\widetilde{IA}$ Theory}

As noted in Section \ref{sec:orient}, the compactification manifold
for the $\widetilde{IA}$ theory is $S^{1,1}$, but with an
$O8^-$-plane at one fixed point and an $O8^+$-plane at the
other. Therefore, $Dp$-branes are classified by
$KR^{p-9}_{(+,-)}(S^{1,1})$. The index is determined as a shift by one
from the $T$-dual theory $KR^{p-8}(S^{0,2})$.  

Recall that
$$S^{1,1}\smallsetminus S^{1,0}\cong \bR^{1,0}\times S^{0,1},$$
with an involution that exchanges the $2$ copies of $\bR$. Therefore
\begin{align}
KR_{\alpha}^{-i}(S^{1,1}\smallsetminus S^{1,0})&\cong
KR^{-i}(\bR^{1,0}\times S^{0,1})\nonumber\\ 
&\cong K^{-i}(\bR)\nonumber\\
&\cong K^{-i-1},
\end{align}
for all $\alpha$. For $\alpha=(+,-)$,
\begin{equation}
KR^{-i}_\alpha(S^{1,0})\cong KO^{-i}\oplus KSp^{-i}.
\end{equation}
Plugging these into equation \eqref{eqn:twKRles} we get the long exact sequence
\begin{equation}
\label{eqn:lespm}
\xymatrix{\cdots\ar[r] & 
  K^{-i-1}\ar[r]& KR_{(+,-)}^{-i}(S^{1,1})\ar[r]& KO^{-i}\oplus
  KSp^{-i}\ar[r]^(.65){\delta}& K^{-i}\ar[r]&\cdots.} 
\end{equation}
The map $\delta$ is complexification on the first summand and doubling on
the second summand, since symplectic bundles contain $2$ complex
bundles. Furthermore, the long exact sequence splits into $2$ parts 
\begin{align}
\xymatrix{0\ar[r] & 
   KR_{(+,-)}^{0\pmod{4}}(S^{1,1})\ar[r]&
   \bZ\oplus\bZ\ar[r]^(.57){\delta}&
   \bZ\ar[r]&KR_{(+,-)}^{-3\pmod{4}}(S^{1,1})\ar[r]&0.}\\ 
   \xymatrix{0\ar[r] & 
   KR_{(+,-)}^{-2\pmod{4}}(S^{1,1})\ar[r]& \bZ_2\ar[r]^(.52){\gamma}&
   \bZ\ar[r]^(.3){\sigma}&KR_{(+,-)}^{-1\pmod{4}}(S^{1,1})\ar[r]&\bZ_2\ar[r] & 0.} 
\end{align}
The map $\delta$ is $(m,n)\mapsto m+2n$, which is surjective. This
means $KR_{(+,-)}^{0\pmod{4}}(S^{1,1})\cong\bZ$ and
$KR_{(+,-)}^{-3\pmod{4}}(S^{1,1})\cong 0$. $\gamma$ must be $0$, showing
$KR_{(+,-)}^{-2\pmod{4}}(S^{1,1})\cong \bZ_2$. This gives us an extension
problem 
\begin{equation}
\label{eq:ext1}
 \xymatrix{0\ar[r] &
   \bZ\ar[r]^(.3){\sigma} &KR_{(+,-)}^{-1\pmod{4}}(S^{1,1})\ar[r]&\bZ_2\ar[r] &
   0.} 
\end{equation}
However, since removing the fixed point with the $-$ sign from
$S^{1,1}_{(+,-)}$ leaves $\bR^{0,1}$, we also have an exact sequence
\[
0 = KSp^{-2} \to KR^{-1}(\bR^{0,1}) \to KR_{(+,-)}^{-1}(S^{1,1}) \to
KSp^{-1}= 0,
\]
and since $KR^{-1}(\bR^{0,1})\cong KO^0 =\bZ$, we see that
$KR_{(+,-)}^{-1}(S^{1,1})\cong\bZ$. The corresponding argument where
we remove the point with the $+$ sign instead shows that
$KR^{-5}_{(+,-)}(S^{1,1})\cong\bZ$.  Putting this all together we find  
\begin{equation}
KR_{(+,-)}^{-i}(S^{1,1})\cong KSC^{-i+1},
\end{equation}
which has periodicity with period $4$.

Now we can see $T$-duality between the type \~{I} and $\widetilde{IA}$
theories as an isomorphism 
\begin{equation}
KR^{-i}(S^{0,2})\cong KR_{(+,-)}^{-i-1}(S^{1,1}).
\end{equation}
The fact that we need to include a charge in the $T$-dual of the IIB
theory on $S^{0,2}$ is contained in the geometry of $S^{0,2}$ in a way
that is explored in \cite{DMDR}. Now let us turn to the different
possibilities of sign choices for $2$-torus orientifolds. 

\subsection{Torus Orientifolds with a Sign Choice}

For a $2$-torus with an involution\footnote{Since this is what's
  needed for physics, we are assuming the torus can be identified with
  a complex smooth curve of genus $1$, and the involution is either
  holomorphic or anti-holomorphic. This is explained further in
  \cite{DMDR}.}  the possible fixed point sets are
empty, $1$ copy of $S^1$, $2$ disjoint copies of $S^1$, $4$ isolated
points, or the entire copy of $T^2$. Obviously, when the fixed set
is empty there are no possible sign choices. Also, when the fixed set
is a single copy of $S^1$ or the entire $2$-torus, then the fixed set
has only a single component. Therefore, there is only one possible
sign choice giving either ordinary $KR$-theory ($+$ sign choice) or an
index shift by $4$ of ordinary $KR$-theory ($-$ sign choice). The only
cases that do not immediately reduce to ordinary $KR$-theory are when
the fixed point set is either $2$ disjoint copies of $S^1$ or $4$
isolated points. 

Let us first consider the orientifold of the $2$-torus with $4$ fixed
points, corresponding to when the involution is
reflection. Topologically our orientifold is 
$S^{1,1}\times S^{1,1}$.  There are $3$ supersymmetric sign choices
for the $4$ fixed points: $\alpha=(+,+,+,+)$, $(+,+,-,-)$, or
$(+,+,+,-)$. (The non-supersymmetric cases $(-,-,-,-)$ 
and $(-,-,-,+)$ can be obtained from $(+,+,+,+)$ and $(+,+,\allowbreak
+,-)$ by an index shift.) 

The case of $(+,+,+,-)$ is considerably more subtle to compute than
the other two, though as shown 
by Witten \cite{Witten:1998-02}, the case of four $O$-planes, three
with a $+$ charge and one with a $-$ charge, does indeed occur in
physics.

\begin{figure}[htb]
\begin{tikzpicture}[scale=2]
\draw [->] (0,-.75) -- (0,.75);
\draw [->] (-.75,0) -- (.75,0);
\fill (0,0) circle (1pt);
\fill (0,.5) circle (1pt);
\fill (.5,0) circle (1pt);
\fill (.5,.5) circle (1pt);
\draw [red, thick] (.25, .25) -- (-.25, .25) -- (-.25, -.25) -- (.25, -.25) 
-- (.25, .25);
\draw [dashed] (-.5,-.5) -- (.5,-.5);
\draw [dashed] (-.5,-.5) -- (-.5,.5);
\draw [dashed] (-.5,.5) -- (.5,.5);
\draw [dashed] (.5,-.5) -- (.5,.5);
\node [below right] at (.5,0) {$\frac{1}{2}$};
\node [above right] at (0,.5) {$\frac{1}{2}$};
\end{tikzpicture}
\caption{Fundamental domain of a $2$-torus with the $4$ fixed points shown.}
\label{Fig:T2fd}
\end{figure}
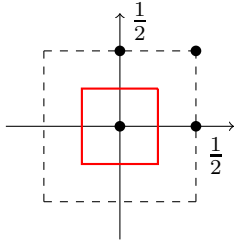

To determine $KR_\alpha$ in each case, we first need the following
result. 
\begin{proposition}
\label{prop:KRminusfixedset}
Let $X$ be $T^2$ {\lp}realized as $\bR^2/\bZ^2${\rp}
with involution given by multiplication by $-1$. Let $Y$ be
the set of $4$ points fixed by the involution. Then 
\begin{equation}
KR_\alpha^{-i}(X\smallsetminus Y)\cong KSC^{-i-1}\oplus K^{-i-1}\oplus K^{-i-1},
\end{equation}
for all $\alpha$.
\end{proposition}
\begin{proof}
A picture of the fundamental domain of the $T^2$ is shown in Figure
\ref{Fig:T2fd}. 
When we remove the four fixed points and the dashed lines along the
boundary of the fundamental domain, what remains retracts onto the
square with vertices at
$\left(\pm\frac{1}{4},\pm\frac{1}{4}\right)$ shown in color. 
One can proceed to compute $KR_\alpha^*$ using this picture, but it
will be faster to use the spectral sequence of Theorem
\ref{thm:chargedBredonSS}. Since $\iota$
acts freely on $X\smallsetminus Y$, the spectral sequence reduces to
the one studied by Karoubi and Weibel \cite[Example A.3]{Karoubi:2005}.
Let $W=(X\smallsetminus Y)/\iota$, which is diffeomorphic to
$S^2\smallsetminus \{4\text{ points}\}$. The map
$(X\smallsetminus Y) \to W$ is a $2$-to-$1$ covering map.
The spectral sequence has $E_2^{p,q}=0$ for $q$ odd, and reduces to
$H^p_c(W, \bZ(i)) \Rightarrow KR^{p+2i}_\alpha(X\smallsetminus Y)$,
where $\bZ(i)=\bZ$ (the constant sheaf) for $i$ even and
$\bZ(i)$ is the nontrivial local coefficient system determined by the
covering map $(X\smallsetminus Y) \to W$ for $i$ odd.
By Poincar\'e duality, $H^p_c(W, \bZ)\cong H_{2-p}(W, \bZ)$, which is
$\bZ$ for $p=2$, $\bZ^3$ for $p=1$, $0$ for $p=0$. The groups
$H^p_c(W, \bZ(1))$ are slightly harder to compute, but can be
obtained, for example, from the exact sequence
\[
\cdots \to H^p_c(\bR\times S^1, \bZ(1)) \to H^p_c(W, \bZ(1)) \to
H^p_c(\bR\amalg \bR, \bZ(1)) \to \cdots,
\]
coming from the fact that deleting two line segments from $W$,
each one running between two of the branch points of the branched
covering $T^2\to S^2$, leaves an open subset diffeomorphic to
$\bR\times S^1$. Here $H^p_c(\bR\amalg \bR, \bZ(1)) \cong
H^p_c(\bR\amalg \bR, \bZ)\cong \bZ^2$ for $p=1$, and $0$ for other values
of $p$, since each component of $\bR\amalg \bR$ is simply connected.
The result is that $H^p_c(W, \bZ(1))$ is isomorphic
to $\bZ^2$ for $p=1$, $\bZ_2$ for $p=2$, and $0$ for other values of
$p$. The spectral sequence is shown in Figure \ref{fig:BSS}.
Note that there is no room for any nontrivial differentials or for any
nontrivial extensions, and the Proposition follows.
\begin{figure}[hbpt]
\[
\xymatrix@R-1.5pc{
q \backslash p & \ar[ddddddd] & 0 & 1 & 2 &\\
\ar[rrrrr] &&&&&\\
0 & &0  & \bZ^3 & \bZ&\\
-1 & &0 & 0 & 0& \\
-2 & &0 & \bZ^2 & \bZ_2& \\
-3 & &0 & 0 & 0&\\
-4 & &0 & \bZ^3 & \bZ&\\
&&&&&}
\]
\caption{$E_2$ of the spectral sequence for computing $KR^*_\alpha(X
  \smallsetminus Y)$. The sequence repeats with vertical period $4$.} 
\label{fig:BSS}
\end{figure}
\end{proof}

For the set of fixed points, $Y$, the three options are
\begin{equation}
KR^{-i}_\alpha(Y) = \left\{
        \begin{array}{ll}
         4KO^{-i},   & \alpha=(+,+,+,+) \\
           2KO^{-i}\oplus2KSp^{-i},   & \alpha=(+,+,-,-) \\
           3KO^{-i}\oplus KSp^{-i},   & \alpha=(+,+,+,-).
          \end{array}
    \right.    
\end{equation}
The case where $\alpha=(+,+,+,+)$ (in the notation of
\cite{Olsen:1999}, this is $T^{1,2}$) just gives ordinary $KR$-theory,
for which we get the calculation
\begin{equation}
\label{KR4fixed}
KR^{-i}(X) \cong KR^{-i}(S^{1,1}) \oplus KR^{-i+1}(S^{1,1})
\cong KO^{-i} \oplus KO^{-i+1} \oplus KO^{-i+1} \oplus KO^{-i+2}.
\end{equation}

The relevant long exact sequence for $\alpha=(+,+,+,-)$ is (via
Proposition \ref{prop:KRminusfixedset})
\begin{multline}
\label{eqn:LES11}
\cdots\to
  KSC^{-i-1}\oplus 2K^{-i-1}\to KR_{(+,+,+,-)}^{-i}(S^{1,1}\times
  S^{1,1})\\ \to 3KO^{-i}\oplus KSp^{-i}\to\cdots.
\end{multline}
This gives an extension problem in determining each of the
$KR$-groups. Therefore, we need to look at some additional long exact
sequences to determine $KR^i_\alpha(S^{1,1}\times S^{1,1})$.  

Let $Y_+=S^{1,1}_{(+,+)}\vee S^{1,1}_{(+,+)}$ be the wedge of $2$ circles
going through the three fixed points with sign choice $+$.
In terms of Figure \ref{Fig:T2fd}, this is the image
of the dotted lines. Then we get a long exact sequence 
\begin{equation}
\label{eqn:LES12}
\xymatrix{\cdots\ar[r] & KR^{-i}(X\smallsetminus Y_+)\ar[r]&
  KR_\alpha^{-i}(X)\ar[r]& KR_\alpha^{-i}(Y_+)\ar[r]&\cdots.}
  \end{equation}
Note that $X\smallsetminus Y_+\cong\bR^{0,2}_-$, where the fixed point
of $\bR^{0,2}$ is given the sign choice $-$. Therefore 
\begin{align*}
KR_\alpha^{-i}(X\smallsetminus Y_+) &\cong KR_-^{-i}(\bR^{0,2})\\
&\cong KSp^{-i+2}.
\end{align*}
To determine $KR^{-i}_{(+,+,+)}(Y_+)$, first note that
this reduces to ordinary $KR$-theory since the sign choices are all
positive. Now consider the split long exact sequence 
\begin{equation}
\label{eq:LESwedgecirc+1}
\xymatrix{\cdots\ar[r] & KR^{-i}(Y_+\smallsetminus \{\pt\})\ar[r]&
  KR^{-i}(Y_+)\ar[r]& KR^{-i}(\pt)\ar[r]\ar@/_/[l]&\cdots,}
\end{equation}
where the basepoint is the joining point of the two circles (a fixed
point with sign $+$). Therefore, $Y_+\smallsetminus\{\pt\}$ is $2$ copies of
$\bR^{0,1}$ and 
$$KR^{-i}(Y_+)\cong KO^{-i+1}\oplus KO^{-i+1}\oplus KO^{-i}.$$
Plugging $KR_\alpha^{-i}(Y_+)$ into the exact sequence
\[
\cdots\to KR^{-i}_\alpha(X\smallsetminus Y_+)\cong KSp^{-i+2} \to
  KR_\alpha^{-i}(X)\to KR_\alpha^{-i}(Y_+)\to \cdots,
\]
we find $KR_{(+,+,+,-)}^{-i}(S^{1,1}\times
S^{1,1})$ is $\bZ$ if $i=4$ or $6$, $\bZ^2$ for $i=5$, and
$\bZ_2^2$ for $i=3$. There are extension problems for the other $4$
indices mod $8$.  

To solve the remaining extension problems, we can repeat the same
process, but use the space $Y_-$ which is the one point union of $2$
circles joined at the fixed point with sign $-$ and going through $2$
of the fixed points with sign choice $+$. This space is the 
image of the coordinate axes in Figure \ref{Fig:T2fd}. Note that
$X\smallsetminus Y_-\cong \bR^{0,2}$ (with a $+$ sign at the fixed point).
If we remove one circle, which we can identify with $S^{1,1}_{(+,-)}$,
from $Y_-$, then what
remains is $\bR^{0,1}$ (with a $+$ sign), so we get an exact sequence
\[
\cdots \to KR^{-i}(\bR^{0,1}) \to KR^{-i}_\alpha(Y_-) \to
KR^{-i}_{(+,-)}(S^{1,1}) \to \cdots,
\]
or in other words,
\begin{equation}
\cdots \to KO^{-i+1} \to
KR^{-i}_\alpha(Y_-)\to KSC^{-i+1} \to KO^{-i+2} \to \cdots.
\label{eq:Yminus}
\end{equation}
In fact \eqref{eq:Yminus} splits, i.e., $KR^{-i}_\alpha(Y_-)\cong 
KO^{-i+1} \oplus KSC^{-i+1}$, since the inclusion $S^{1,1}_{(+,-)}
\hookrightarrow Y_-$ is split by the (sign-preserving) ``fold map''
sending both circles in $Y_-$ onto $S^{1,1}_{(+,-)}$.
Putting our result for $Y_-$ into the exact sequence
\[
\cdots \to KR^{-i}(\bR^{0,2})\cong KO^{-i+2} \to KR^{-i}_\alpha(X) \to 
KR^{-i}_\alpha(Y_-) \to KO^{-i+3} \to \cdots
\]
gives that $KR_{(+,+,+,-)}^{-i}(S^{1,1}\times
S^{1,1})$ is $\bZ$ for $i=0$, $\bZ^2$ for $i=1$, $\bZ\oplus(\bZ_2)^2$ for
$i=2$ (for this case we must combine the information we get
from \eqref{eqn:LES11},  \eqref{eq:LESwedgecirc+1}, and
\eqref{eq:Yminus}), and $0$ for $i=7$. The results of the calculation
are summarized in the last column of Table \ref{table:twistedKO} in
Section \ref {sec:KOtwist} below.   

We could also use the spectral sequence in Theorem
\ref{thm:chargedBredonSS} for $S^{1,1}\times S^{1,1}$ with
$\alpha=(+,+,+,-)$. To determine the $E_2$ term, we need to look at
the groups 
$H^p_G(X;\uwave{KR_\alpha}^{q})$. These are most easily computed
using the exact sequence
\begin{equation}
\cdots \to H^p_{G,c}(X\smallsetminus X^\iota;\uwave{KR_\alpha}^q)
\to H^p_G(X;\uwave{KR_\alpha}^q) \to
H^p_G(X^\iota;\uwave{KR_\alpha}^q)\to \cdots.
\label{eq:KRalphatwists}
\end{equation}
Here $H^p_G(X^\iota;\uwave{KR_\alpha}^q)$ is non-zero only for $p=0$,
where it is $3KO^q\oplus KSp^q$, and $H^p_{G,c}(X\smallsetminus X^\iota;
\uwave{KR_\alpha}^q)$ was computed in the proof of
Proposition \ref{prop:KRminusfixedset}.
In \eqref{eq:KRalphatwists} there is one potentially nonzero
connecting map, $\bZ^4\cong 3KO^q\oplus KSp^q \to \bZ^3$ when $q\equiv0\pmod4$.
This map can be computed by comparison with the corresponding sequences
for the cases of $S^{1,1}_{+,+}$ and $S^{1,1}_{+,-}$, where the
$KR_\alpha$ groups were computed from equation \eqref{eqn:lespm} and the
surrounding discussion. One finds that the connecting map
has kernel $\bZ$ in all cases, is surjective for $q\equiv0\pmod8$, and
has a cokernel of $\bZ_2^2$ when $q\equiv4\pmod8$. Thus the groups
$H^p_G(X^\iota;\uwave{KR_\alpha}^q)$ are as in Figure
\ref{fig:cBredonSSpppm}. This calculation is consistent with our
computation of $KR^*_{(+,+,+,-)}(S^{1,1}\times S^{1,1})$, assuming
that there are $d_2$ differentials that kill off the $\bZ_2$'s in
positions $(2, -2)$ and $(2, -6)$. 

\begin{figure}[h!]
\[
\xymatrix@R-1.5pc{
q \backslash p & \ar[dddddddddd] & 0 & 1 & 2 &\\
\ar[rrrrr] &&&&&\\
0 & &\bZ & 0 & \bZ&\\
-1 & & \left(\bZ_2\right)^3 & 0 & 0&\\
-2 & & \left(\bZ_2\right)^3 & \bZ^2 & \bZ_2&\\
-3 & &0 & 0 & 0&\\
-4 & & \bZ & \left(\bZ_2\right)^2 & \bZ&\\
-5 & & \bZ_2 & 0 & 0&\\
-6 & & \bZ_2 & \bZ^2 & \bZ_2&\\
-7 & &0 & 0 & 0&\\
&&&&&}
\]
\caption{$E_2$ of the spectral sequence for computing
  $KR_{(+,+,+,-)}^*(S^{1,1}\times S^{1,1})$. The sequence repeats
  with vertical period $8$.}  
\label{fig:cBredonSSpppm}
\end{figure}

The case $\alpha=(+,+,-,-)$ can be obtained by the product of the type
$\widetilde{IA}$ theory with itself or the type IA theory. The 
equivariant decomposition 
\[
S^{1,1}_{(+,-)}\times S^{1,1} = (S^{1,1}_{(+,-)}\times \{\pt\})\amalg
(S^{1,1}_{(+,-)}\times \bR^{0,1})
\]
gives the calculation
\begin{equation}
KR^{i}_{(+,+,-,-)}(S^{1,1}\times S^{1,1})\cong KSC^{i+2}\oplus
KSC^{i+1}.
\label{eq:2plus2minus}
\end{equation}

The same case can also be obtained by looking at
\[
S^{1,1}_{(+,-)}\times S^{1,1}_{(+,-)} = (S^{1,1}_{(+,-)}\times \{\pt\})\amalg
(S^{1,1}_{(+,-)}\times \bR^{0,1}_-).
\]
But crossing with $\bR^{0,1}_-$ has the same effect as crossing with
$\bR^{4,1}$ or with $\bR^{3,0}$, and since $KSC^*$ is $4$-periodic, we get
the same result as in \eqref{eq:2plus2minus}.

Now let us consider orientifolds of the $2$-torus where the fixed set
is $2$ disjoint copies of $S^1$. Topologically, this is $S^{1,1}\times
S^{2,0}$. There are $2$ possible supersymmetric sign choices, $(+,+)$
and $(+,-)$. As usual, the non-supersymmetric case $(-,-)$ can be
obtained from $(+,+)$ by an index shift. When both fixed circles have
sign choice $+$, $KR_\alpha$ reduces to ordinary $KR$-theory, 
\begin{align}
KR^{-i}(S^{1,1}\times S^{2,0}) &\cong KR^{-i-1}(S^{1,1})\oplus KR^{-i}(S^{1,1})\nonumber\\
&\cong KO^{-i-1}\oplus KO^{-i}\oplus KO^{-i}\oplus KO^{-i+1}.
\end{align}

The case $\alpha=(+,-)$ is just the product of the type
$\widetilde{IA}$ theory, $S^{1,1}_{(+,-)}$, with a fixed circle,
$S^{2,0}$, so we find 
\begin{align}
KR^{-i}_{(+,-)}(S^{1,1}\times S^{2,0}) &\cong
KR_{(+,-)}^{-i-1}(S^{1,1})\oplus KR_{(+,-)}^{-i}(S^{1,1})\nonumber\\ 
&\cong KSC^{-i}\oplus KSC^{-i+1}.
\end{align}

To conclude this section, we explain how to compute $KR$-theory for a
$2$-torus orientifold where the involution $\iota$ is orientation
reversing and has a fixed set that is topologically $S^1$. Unlike the
cases above, this orientifold does \emph{not} split as a product of
two circle orientifolds, so a somewhat more complicated calculation is
required. 

\begin{theorem}
\label{thm:KRspecies1}
Let $(X, \iota)$ be a Real space where $X=T^2$ and $\iota$ is smooth,
orientation reversing, and has a fixed set that is topologically
$S^1$. The quotient space $M=X/\iota$ is topologically a closed
M{\"o}bius strip.
{\lp}Such a space arises from taking $X$ to be the complex points of a
smooth projective real curve of genus $1$ when the real points have
exactly one connected component, and taking $\iota$ to be the action
of $\operatorname{Gal}(\bC/\bR)$.{\rp} Then
$KR^j(X, \iota) \cong \left(KO^j\right)^2 \oplus KU^{j-1}$.
\end{theorem}
\begin{proof}
\emph{Step 1.}
Since $X$ has a nonempty fixed set,
$KR^{-j}(T^2,\iota) \cong \tKR^{-j}(T^2,\iota) \oplus KO^{-j}$, and we
only need to compute $\tKR^{-j}(T^2,\iota)$.
We begin by deducing two useful exact sequences. The first comes from
observing that $X\smallsetminus X^\iota\cong S^{0,2}\times \bR^{0,1}$
(as a Real space). Thus $KR^{-j}\left(X\smallsetminus X^\iota\right)
\cong KR^{-j+1}(S^{0,2})\cong KSC^{-j+1}$. Since
$\tKR^{-j}(X^\iota)\cong KR^{-j}(\bR^{1,0})\cong KO^{-j-1}$, we get the long exact
sequence
\begin{equation}
\label{eq:species1KSC}
\cdots \to KO^{-j-2} \xrightarrow{\delta} KSC^{-j+1} \to \tKR^{-j}(T^2,
\iota) \to  KO^{-j-1} \xrightarrow{\delta} KSC^{-j+2} \to \cdots,
\end{equation}
where the connecting map $\delta$ will be determined later. However,
note for now that $\delta$ vanishes after inverting $2$, since
$KO^{-j-1}\left[\frac12\right]$ is nonzero only for $j\equiv 3\pmod{4}$
and $KSC^{-j+2} \cong \bZ_2$ for these values of $j$. Thus the
torsion-free part of $\tKR^{-j}(X, \iota)$ is the same as for
$KSC^{-j+1} \oplus KO^{-j-1}$ and is thus $\bZ$ for $j$ odd, $\bZ$ for
$j\equiv 0\pmod{4}$, and $0$ for $j\equiv 2\pmod{4}$.

To get the other exact sequence, choose an interval $I$ in $M=
X/\iota$ transverse to the central circle and meeting the boundary
in two points. The inverse image of this interval in $X$ is a copy of
$S^{1,1}$, the unit circle in the complex plane with complex
conjugation as the involution. Furthermore the complement of this copy
of $S^{1,1}$ is isomorphic (as a Real space) to $(0,1)\times
S^{1,1}$. Since $S^{1,1}$ with one fixed point removed is isomorphic
(as a Real space) to $\bR^{0,1}$, $\tKR^j(S^{1,1})\cong KR^j(\bR^{0,1})
\cong KO^{j+1}$ via \cite[Theorem 2.3]{MR0206940} and
$KR^j((0,1)\times S^{1,1})\cong KO^{j-1}\oplus KO^j$. So we get an exact
sequence
\begin{equation}
\cdots \to KO^j \xrightarrow{\rho} 
KO^{j-1}\oplus KO^j\to \tKR^j(X) \to KO^{j+1} \xrightarrow{\rho} KO^{j} \oplus
KO^{j+1} \to \cdots.
\label{eq:KRbycutting}
\end{equation}

\emph{Step 2.}
Observe next that the connecting maps $\delta$ and $\rho$
have to be compatible with cup products by the ground ring 
\[
KO^* \cong \bZ[b^{\pm}, \xi,
  \eta]/(2\eta, \eta^3, \xi\eta, \xi^2-4b).
\]
Here the torsion-free generators are $b$ in degree $-8$ and $\xi$ in
degree $-4$, and the torsion generator $\eta$ is in degree $-1$.
To prove this claim, simply replace $T^2$ by $T^2\times \bR^{p,q}$.
Thus the connecting map $\delta:KO^{j-1} \to KSC^{j+2}$ has to be of the form $x\mapsto x\cdot y$, $x\in KO^*$ and $y$ some class in
$KSC^3$, and the connecting map $\rho:KO^{j} \xrightarrow{\rho}
KO^{j-1} \oplus KO^{j}$ has to be of the form $x\mapsto (ax, bx)$,
where $a\in KO^{-1}$ and $b\in KO^0$.

\emph{Step 3.}
The ground ring for $KSC$ theory is
\[
KSC^* \cong \bZ[\beta^{\pm}, \eta']/(2\eta', \eta'^2),
\]
where the periodicity element $\beta$ is in degree $-4$ and the
torsion generator $\eta'$ is in degree $-1$. There is a 
canonical ring homomorphism $\varepsilon\co KO^*\to KSC^*$
(the map on $KR$ induced by $S^{1,1}\to
\hbox{pt}$). Then $\varepsilon(\eta)=\eta'$, $\varepsilon(b)=\beta^2$,
and $\varepsilon(\xi)=2\beta$.  These are standard facts which can be
found in \cite[\S1]{MR1935138}, for example.

\emph{Step 4.}
We claim that $\delta$ is given by $x\mapsto \varepsilon(x)\cdot \beta^{-1}
\eta'$ and that $\rho$ is given by $x\mapsto (x\cdot \eta, 0)$.
We get this by playing off the sequences \eqref{eq:species1KSC} and \eqref{eq:KRbycutting} against each
other. Start with $\delta(x) = \varepsilon(x)\cdot y$, $y\in KSC^3$.
If $y$ were $0$, we'd have a short exact sequence
\[
0\to KSC^{j+1} \to \tKR^j(T^2,
\iota) \to  KO^{j-1} \to 0,
\]
and this would imply for example that $\tKR^{-6}(T^2,\iota) \cong
KSC^{-5} \cong
\bZ_2$, which contradicts what we obtain from the other exact sequence
\eqref{eq:KRbycutting} for $j=-6$. Thus $y = \beta^{-1}\eta'$ (the
generator of $KSC^3$) and the claim follows.

Recall that $\rho$ is of the form $x\mapsto (ax, bx)$ with $a\in
KO^{-1}$ and $b\in KO^0\cong \bZ$. The number $b$ must be $0$;
otherwise the torsion-free part of $\tKR^j(T^2,\iota)$ would
contradict what we got in Step 1 from \eqref{eq:species1KSC}.  And $a\in
KO^{-1}$ can't vanish, 
because if it did, we'd have a short exact sequence
\[
0 \to KO^{-2} \oplus KO^{-1} \to \tKR^{-1}(T^2,
\iota) \to  KO^{0} \to 0,
\]
giving $\tKR^{-1}(T^2,\iota) \cong \bZ_2^2 \oplus \bZ$,
while \eqref{eq:species1KSC} gives that $\tKR^{-1}(T^2,\iota)$ is
either $\bZ$ or $\bZ \oplus \bZ_2$.  So this completes the calculation of 
the boundary maps $\delta$ and $\rho$.

\emph{Step 5.}
To conclude, we use a well-known fact in homotopy 
theory \cite[p.\ 206]{MR0402720}, which is
that if $\mathbf{KU}$ and $\mathbf{KO}$ are the complex and real
topological $K$-theory spectra, then there is a fiber/cofiber
sequence of spectra
\[
\Sigma \mathbf{KO} \xrightarrow{\eta} \mathbf{KO} \xrightarrow{c} 
\mathbf{KU}. 
\]
This corresponds to a famous long exact sequence
\cite[Theorem III.5.18]{MR0488029} or \cite[Definition 1.13(2)]{MR1935138}:
\[
\cdots \to KO^{-n}(X) \xrightarrow{\eta} KO^{-n-1}(X) \xrightarrow{c}
KU^{-n-1}(X)  \xrightarrow{r\beta_U^{-1}} KO^{-n+1}(X) \to \cdots.
\]
Here $c$ is complexification, $r$ is realification, and $\beta_U$ is
the complex Bott element.

Because of our calculation of the boundary map $\rho$, $KO^j$ splits
off as a direct summand in $\tKR^{-j}(T^2,\iota)$, and the complement
can be identified with the cofiber of $\eta$ with a degree shift. So
this completes the proof.
\end{proof}

We conclude by noting that \cite[Theorem 4.8]{Karoubi:2005} says that
if $X$ is a smooth projective variety defined over $\bR$ (which in our
case will be a curve of genus $1$), identified with the Real space of
its complex points with involution given by the action of
$\operatorname{Gal}(\bC/\bR)$, then the natural map 
$K_j(X;\bZ_2) \to KR^{-j}(X;\bZ_2)$ sending algebraic to topological
$K$-theory is an isomorphism for $j$ sufficiently large (in our case
$j\ge 1$ suffices). Here algebraic $K$-theory or $KR$-theory with
$\bZ_2$ coefficients is related to the integral theory by a universal
coefficient sequence
\begin{equation}
0 \to KR^{-j}(X)/2 \to KR^{-j}(X;\bZ_2) \to {}_2KR^{-j+1}(X)\to 0,
\label{eq:UCT}
\end{equation}
where ${}_2KR^{-j+1}(X;\bZ_2)$ denotes the $2$-torsion in
$KR^{-j+1}(X)$, and similarly for $K_j$. The torsion subgroup
of $K_j(X)$ was computed in \cite[Main Theorem 0.1]{MR1936583} and
agrees with our results under this isomorphism.\footnote{There is a
  small typo in the statement of \cite[Main Theorem
  0.1]{MR1936583}. $K_2(X)_{\text{tors}}$ should contain $\nu+1$ copies of
  $\bZ_2$ (here $\nu$ is the species), not $\nu$ copies as written.}

\section{More general twists and why they are needed for physics}
\label{sec:generaltwists}

\subsection{Twisted $KO$-theory}
\label{sec:KOtwist}

While twisted complex $K$-theory is by now well-known in both the
mathematics literature (e.g., \cite{MR0282363,MR1018964,MR2172633,
MR2307274,MR2513335}) and the physics literature (e.g.,
\cite{MR1827946,MR2080959}),
its cousin, twisted real $K$-theory, is defined similarly but is less
familiar. One way to define it is by using the $K$-theory of real
continuous-trace algebras of real type (see \cite[\S3]{MR1018964}). In
the separable case, after stabilization, such an algebra is the
algebra of sections vanishing at infinity of a bundle whose fibers are the
compact operators $\cK_{\bR}$ on an infinite-dimensional separable
real Hilbert space $\cH_{\bR}$. Since $O(\cH_{\bR})$ is contractible
but the automorphism group of $\cK_{\bR}$ is the \emph{projective} orthogonal
group $PO(\cH_{\bR}) = O(\cH_{\bR})/\bZ_2$, which is a $K(\bZ_2, 1)$
space, the relevant algebra bundles are classified by homotopy classes of maps 
from the space $X$ to $BPO(\cH_{\bR})$, which 
is a $K(\bZ_2, 2)$ space. Thus they are classified by a single
characteristic class $\widetilde w_2 \in H^2(X, \bZ_2)$, which one can
identify with the characteristic class for
Witten's type I string theory without vector structure in
\cite{Witten:1998-02}.  In other 
words, for each $\widetilde w_2 \in H^2(X, \bZ_2)$, one gets an
$8$-periodic family of $K$-groups $KO^*(X, \widetilde w_2)$, reducing
to $KO^*(X)$ when $\widetilde w_2 =0$.  This is analogous to twisting
by $H$-flux for complex $K$-theory. Recall that the
automorphism group of $\cK$ (the compact operators on a \emph{complex}
infinite-dimensional separable Hilbert space $\cH$) is the projective
unitary group 
$PU(\cH)=U(\cH)/S^1$. In this case the relevant algebra bundles are
classified by homotopy classes of maps from $X$ to $BPU(\cH)$, which
is a $K(\bZ,3)$ space. Therefore they are classified by the $H$-flux
$H\in H^3(X;\bZ)$. 

Just as in the complex case, the twisted real $K$-theory groups can be
computed using an Atiyah-Hirzebruch spectral sequence (AHSS)
\[
H^p_c(X, KO^q) \Rightarrow KO^{p+q}(X, \widetilde w_2 ),
\]
where $H^*_c$ is cohomology with compact supports and
$\widetilde w_2$ appears in the differentials. We will primarily be
interested in the case $X=T^2$, in which case the ``compact supports''
modifier can be dropped and there is only room for one differential,
\[
d_2\co H^0(T^2, KO^q) = KO^q \to KO^{q-1} \cong H^2(T^2, KO^{q-1}).
\]
This differential is cup product with $\widetilde w_2$, viewed
as an element of $H^2(T^2, KO^{-1})\cong \bZ_2$.  So if $\widetilde
w_2$ is the nontrivial element of $H^2(T^2, \bZ_2)$, the $E_2$ term of
the spectral sequence with the non-zero $d_2$ differentials indicated
is shown in Figure \ref{fig:AHSS}, and the $E_3=E_\infty$ term is
shown in Figure \ref{fig:AHSSe3}.
\begin{figure}[hbpt]
\[
\xymatrix@R-1.5pc{
q \backslash p & \ar[dddddddddd] & 0 & 1 & 2 &\\
\ar[rrrrr] &&&&&\\
0 & &\bZ \ar[rrd] & \bZ^2 & \bZ&\\
-1 & &\bZ_2 \ar[rrd] & \bZ_2^2 & \bZ_2& \\
-2 & &\bZ_2 & \bZ_2^2 & \bZ_2& \\
-3 & &0 & 0 & 0&\\
-4 & &\bZ & \bZ^2 & \bZ&\\
-5 & &0 & 0 & 0&\\
-6 & &0 & 0 & 0&\\
-7 & &0 & 0 & 0&\\
&&&&&}
\]
\caption{$E_2$ of the spectral sequence for computing $KO^*(T^2,
  \widetilde w_2)$. The sequence repeats with vertical period $8$.} 
\label{fig:AHSS}
\end{figure}
\begin{figure}[thbp]
\[
\xymatrix@R-1.7pc{
q \backslash p & \ar[dddddddddd] & 0 & 1 & 2 &\\
\ar[rrrrr] &&&&&\\
0 & &\bZ & \bZ^2 & \bZ&\\
-1 & & 0 & \bZ_2^2 & 0&\\
-2 & &\bZ_2 & \bZ_2^2 & 0&\\
-3 & &0 & 0 & 0&\\
-4 & &\bZ & \bZ^2 & \bZ&\\
-5 & &0 & 0 & 0&\\
-6 & &0 & 0 & 0&\\
-7 & &0 & 0 & 0&\\
&&&&&}
\]
\caption{$E_\infty$ of the spectral sequence for computing $KO^*(T^2,
  \widetilde w_2)$. The sequence repeats with vertical period $8$.} 
\label{fig:AHSSe3}
\end{figure}

The groups $KO^*(T^2, \widetilde w_2 )$ are thus determined up to
extensions by summing along the diagonals (where $p+q$ takes a
constant value). We see that $KO^0(T^2, \widetilde w_2 )$ is an
extension of $\bZ$ by $\bZ_2^2$, necessarily split, $KO^{-2}(T^2,
\widetilde w_2 )$ is an extension of $\bZ_2$ by $\bZ$, and the
remaining groups  
$KO^j(T^2, \widetilde w_2 )$ are $\bZ_2^2$ for $j=-1$, $\bZ^2$ for
$j=-3$, $\bZ$ for $j=-4$, $0$ for $j=-5$, $\bZ$ for $j=-6$, $\bZ^2$
for $j=-7$. The only case where we are left with an extension problem
is $j=-2$.  It turns out that $KO^{-2}(T^2, \widetilde w_2 )\cong
\bZ$, which we can see as follows. A map of degree one $T^2\to S^2$
collapsing the $1$-skeleton $S^1\vee S^1$ to a point induces a map of
spectral sequences which is an isomorphism on the columns with $p=0$
and $p=2$, hence shows that $KO^{-2}(T^2, \widetilde w_2 )\cong
KO^{-2}(S^2, \widetilde w_2 )$ (with a non-trivial twist in $H^2(S^2,
\bZ_2) \cong \bZ_2$), so we only need to compute this latter group and
show that it is torsion-free.
This will be done in Section \ref{sec:KOtwistH1} below.

There are many ways of seeing that this sort of $\widetilde w_2$
twisting of $KO$ is needed for D-brane classification in the
``no vector structure'' theory of \cite{Witten:1998-02}.  But the key
feature is that Chan-Paton bundles are given not by $O(n)$ bundles
but by $PO(n)$ bundles \cite[\S2.1]{Witten:1998-02},
\cite[\S7.2]{Gao:2010ava}, which is precisely how our twisting was
defined. 

The physics literature suggests that there should be a 
T-duality between the ``type I with no vector
structure'' theory on $T^2$ and the type IIA orientifold on an
elliptic curve with antiholomorphic involution of species $1$ (i.e.,
a fixed set which is topologically just a single circle)
\cite{Keurentjes:2000}.  The D-brane charges in this
theory are described by the groups $KR^j(T^2, \iota)$, where $\iota$
is an involution on $T^2$ with fixed set $S^1$. These groups were
computed above in Theorem \ref{thm:KRspecies1}. Table
\ref{table:twistedKO} shows $KO^j(T^2, \widetilde w_2)$ for
$\widetilde w_2\ne 0$, $KR^j(T^2, \iota)$ for the species $1$
antiholomorphic involution $\iota$, and
$KR^j_{(+,+,+,-)}(S^{1,1}\times S^{1,1})$ from Section
\ref{sec:chargedKR}. The second column agrees  
precisely with the first column shifted down by $1$, 
and the third column agrees with the second column shifted down by $1$, 
as is predicted by T-duality. Note that the data of the $B$-field for
the type IIB theory on $S^{1,1}\times S^{1,1}$ with $\alpha=(+,+,+,-)$
is encoded in the non-triviality of the $d_2$ differential for the
spectral sequence in Figure \ref{fig:cBredonSSpppm}. In \cite{DMDR} we
will describe how the $B$-field is described by a sign choice under
$T$-duality. 

\begin{table}[htb]
\begin{center}
\begin{tabular}{||c||c||c||c||}
\hline
$j$ mod $8$ & $KO^j(T^2,\widetilde w_2)$ & $KR^j(T^2, \iota)$
& $KR^j_{(+,+,+,-)}(S^{1,1}\times S^{1,1})$\\
\hline\hline
$0$ &$\bZ\oplus\bZ_2^2$ &$\bZ^2$ & $\bZ$\\
$-1$&$\bZ_2^2$ &$\bZ\oplus\bZ_2^2$ & $\bZ^2$ \\
$-2$&$\bZ$ &$\bZ_2^2$ & $\bZ\oplus\bZ_2^2$\\
$-3$&$\bZ^2$ &$\bZ$ & $\bZ_2^2$\\
$-4$&$\bZ$ &$\bZ^2$& $\bZ$ \\
$-5$&$0$ &$\bZ$ & $\bZ^2$ \\
$-6$&$\bZ$ &$0$ & $\bZ$ \\
$-7$&$\bZ^2$ &$\bZ$  & $0$\\
\hline\hline
\end{tabular}
\caption{$KO^j(T^2,\widetilde w_2)$, $KR^j(T^2, \iota)$, and
  $KR^j_{(+,+,+,-)}(S^{1,1}\times S^{1,1})$}  
\label{table:twistedKO}
\end{center}
\end{table}

\subsection{Twisted $KO$-theory with an $H^1$ twist}
\label{sec:KOtwistH1}

Twisting of $KO^*(X)$ by $H^1(X$, $\bZ_2)\times H^2(X,\bZ_2)$ was
already defined by Donovan and Karoubi in \cite{MR0282363}. 
(The group of twists $HO(X)$ is actually
a non-split abelian extension of $H^2(X,\bZ_2)$ by
$H^1(X,\bZ_2)$.) For $X$ compact and $\cA$ a bundle over $X$ whose
fibers are $\bZ_2$-graded simple $\bR$-algebras, with
$w(\cA)=\alpha\in HO(X)$, $KO^\alpha(X)$
is the Grothendieck group of graded real vector bundles $X$ which are
finitely generated projective modules for $\cA$.
Here $w(\cA)=\bigl(w_1(\cA),w_2(\cA)\bigr)$, where 
vanishing of $w_2(\cA)\in H^2(X,\bZ_2)$
is the condition for $\cA$ to be the endomorphism bundle of a
$\bZ_2$-graded vector bundle, and $w_1(\cA)=w_1(V)$ if $\cA$ is the
Clifford algebra bundle of a real vector bundle $V$ for a negative
definite metric \cite[Lemma 7]{MR0282363}. When $w_1=0$, we get
back the twisted $KO$-groups of Section \ref{sec:KOtwist}.
The basic composition rule in $HO(X)$ is that 
\[
w_1(\cA\hat\otimes \cB) = w_1(\cA) + w_1(\cB),\qquad
w_2(\cA\hat\otimes \cB) = w_2(\cA) + w_2(\cB) + w_1(\cA)\cdot
w_1(\cB).
\]
For general $X$, $HO(X)$ can have elements of order $4$, but this
won't happen if (as for $S^1$ or $T^2$) every element of $H^1(X,\bZ_2)$
has square $0$. Thus (assuming this condition) every element of
$HO(X)$ is its own inverse, and by the Thom Isomorphism Theorem of
\cite[\S6]{MR0282363}, if $V$ is a real vector bundle over $X$, 
\begin{equation}
KO^j(V)\cong KO^{j-\dim V} (X, w_1(V),w_2(V)).
\label{eq:twistedThomKO}
\end{equation}

As an example of \eqref{eq:twistedThomKO}, 
we can compute $KO^j(S^1, w_1)$ for the nontrivial
element $w_1\in H^1(S^1,\bZ_2)\cong \bZ_2$. Indeed, we have
\[
KO^j(S^1, w_1)\cong KO^{j+1}(V) \cong \widetilde{KO}^{j+1}(\bR\bP^2),
\]
where $V$ is the nontrivial real line bundle over $S^1$, that is, the
M\"obius strip. And the $KO$-groups of $\bR\bP^2$ were computed in
\cite[Theorem 1]{MR0219060}. The result is given in Table
\ref{table:twistedS1}. Here the surprise is the existence of
$4$-torsion in $\widetilde{KO}^0(\bR\bP^2) \cong 
KO^{-1}(S^1, w_1)$.

\begin{table}[htb]
\begin{center}
\begin{tabular}{||c||cccccccc||}
\hline\hline
$j$&$0$&$-1$&$-2$&$-3$&$-4$&$-5$&$-6$&$-7$\\
\hline
$KO^j(S^1,w_1)$\rule{0pt}{12pt} 
& $\bZ_2$& $\bZ_4$& $\bZ_2$ & $\bZ_2$&$0$&$0$&$0$& $\bZ_2$  \\
\hline\hline
\end{tabular}
\caption{$KO^j(S^1,w_1)$ for the nontrivial twist} 
\label{table:twistedS1}
\end{center}
\end{table}

The groups in Table \ref{table:twistedS1} can once again be explained
by a twisted Atiyah-Hirzebruch spectral sequence with starting point
$H^p(S^1, KO^q)$, but this time the only differential is $d_1$, which
is multiplication by $2$ in the places indicated by the arrows in
Figure \ref{fig:AHSSd1}.

\begin{figure}[hbpt]
\[
\xymatrix@R-1.5pc{
q \backslash p & \ar[dddddddddd] & 0 & 1 & \\
\ar[rrrr] &&&&\\
0 & &\bZ \ar[r] & \bZ & \\
-1 & &\bZ_2  & \bZ_2 & \\
-2 & &\bZ_2 & \bZ_2 & \\
-3 & &0 & 0 & \\
-4 & &\bZ \ar[r] & \bZ & \\
-5 & &0 & 0 & \\
-6 & &0 & 0 & \\
-7 & &0 & 0 & \\
&&&&}
\]
\caption{$E_1$ of the spectral sequence for computing $KO^*(S^1,
  w_1)$. The sequence repeats with vertical period $8$.  Arrows
  represent multiplication by $2$, so in $E_2=E_\infty$, each $\bZ$ in the
$p=1$ column is replaced by a $\bZ_2$, and each $\bZ$ in the
$p=0$ column dies.} 
\label{fig:AHSSd1}
\end{figure}

The calculation of $KO^*(S^1, w_1)$ also enables us to compute
$KO^*(T^2, w_1)$  for any choice of a twisting $w_1\in H^1(T^2,
\bZ_2)$. The reason is that for any such $w_1\ne 0$, we can choose a
topological splitting $T^2=S^1\times S^1$ with respect to which $w_1$
lives only on the first factor, so that $(T^2, w_1)\cong (S^1, w_1)
\times (S^1, 0)$. It follows that $KO^j(T^2, w_1)$ splits as
$KO^j(S^1, w_1)\oplus KO^{j-1}(S^1, w_1)$.

The calculation of $KO^*(S^1, w_1)$ also enables us to compute
$KO^\alpha(T^2)$ in the sense of Donovan-Karoubi for a twist $\alpha$
with both $w_1(\alpha)$ and $w_2(\alpha)$ nonzero. Indeed, let $V$
again be the nontrivial real line bundle over $S^1$, that is, the
M\"obius strip. Then $V\times V$ (the Cartesian product) is a rank-two
real vector bundle over $S^1\times S^1=T^2$. If $a$ and $b$ are the
elements of $H^1(T^2, \bZ_2)$ dual to the two circles in the
decomposition $T^2=S^1\times S^1$, then $V\times V$ can be identified
with the Whitney sum $L_a\oplus L_b$, since the fiber of $V\times V$
over $(x,y)\in S^1\times S^1$ is $L_a(x,y) \times L_b(x,y) = L_a(x,y)
\oplus L_b(x,y)$. Note that $w_1(L_a\oplus L_b) = a+b$ and
$w_2(L_a\oplus L_b) = ab$, a generator of $H^2(T^2, \bZ_2)$. So
$KO^j(T^2, a+b, ab) = KO^{j+2}(V\times V)$. The same holds for 
$KO^j(T^2, w_1, w_2)$ for any nonzero $w_1$, $w_2$ since there is a
self-homeomorphism of $T^2$ sending $w_1$ to $a+b$.  Finally we can
compute $KO^j(T^2, w_1, w_2)\cong KO^{j+2}(V\times V)$ using
the fact that $V$ has a closed subspace homeomorphic to $\bR$, with
$(V\smallsetminus \bR)\cong \bR^2$, so that we get from the pair
$(V\times V, V\times \bR)$ an exact sequence
\[
\cdots \to KO^j(V) \to KO^j(V) \to
KO^j(T^2, w_1, w_2) \to KO^{j+1}(V) \to \cdots .
\]
In particular, $KO^j(T^2, w_1, w_2)$ is a $2$-primary torsion group
for all $j$. 

Finally, we mention still another application of the Thom isomorphism
\eqref{eq:twistedThomKO}, namely the completion of the calculation of
$KO^*(T^2, \widetilde w_2)$ when the twist is nonzero. Observe that
the nonzero element $\widetilde w_2\in H^2(T^2,\bZ_2)$ is pulled back
from the generator of $H^2(S^2,\bZ_2)$ under a map $T^2\to S^2$ of
degree one, so to compute $KO^*(T^2, \widetilde w_2)$, we can begin by
computing $KO^*(S^2, \widetilde w_2)$. The generator of
$H^2(S^2,\bZ_2)$ is $\widetilde w_2$ for the underlying real $2$-plane bundle of
the Hopf (complex) line bundle over $S^2\cong \bC\bP^1$, for which the
total space is $\bC\bP^2\smallsetminus\{\pt\}$. So by
\eqref{eq:twistedThomKO}, $KO^{-j}(S^2, \widetilde w_2) \cong
\widetilde{KO}^{-j+2}(\bC\bP^2)$, which is computed in \cite[Theorem
2]{MR0219060}. (The degree $0$ part was computed earlier in
\cite[\S3.6]{MR0165532}.)
Rather surprisingly, $\widetilde{KO}^*(\bC\bP^2)$ is
entirely torsion-free, with copies of $\bZ$ in all even degrees and
nothing in odd degrees. Thus $KO^{-2}(S^2, \widetilde w_2) \cong
KO^{-2}(T^2, \widetilde w_2)\cong
\widetilde{KO}^{0}(\bC\bP^2)\cong \bZ $, not $\bZ\oplus \bZ_2$. 

We should mention that even though we didn't need it for studying
D-brane charges in orientifold theories on $2$-tori, in higher dimensional
situations one might be forced to consider all the various kinds of
twists of $KR$ (sign choice, $H^1$, and $H^2$) simultaneously.  The
general framework for such twists is included in the work of Moutuou
\cite{2011arXiv1110.6836M,MR3158706}. 

\section{Conclusion}

$KR$-theory with a sign choice (Definition \ref{def:Krsign}) allows us
to give a mathematical description of $D$-brane charges for all
orientifolds including ones with both $O^+$- and $O^-$-planes. The
additional data of a sign choice is required to distinguish between
topologically equivalent spaces with different $O$-plane content. As
we saw, $KR$-theory with a sign choice gives a purely mathematical
description of the $D$-branes in the type $\widetilde{IA}$
theory. This calculation provides further evidence for $T$-duality
rather than requiring its assumption to determine the brane charges. 

In addition to providing new tests of $T$-duality, $KR$-theory with a
sign choice predicts the $D$-brane content in theories that could not
be computed previously (which in turn can aid in the discovery of
unknown dualities). We are not aware of the $D$-brane content for the
type I theory without vector structure or either of its $T$-dual
theories appearing anywhere in the literature. This extends the
usefulness of $K$-theory as a first check for $D$-brane content to
orientifold theories. As noted previously, the $K$-theoretic
description cannot determine the sources for the $D$-brane charge,
only that there is a stable charge. Determining the stable charges
using $KR$-theory with a sign choice can greatly constrain what
sources need to be tested for stability at different points in the
moduli space using other methods (such as considering the boundary state
description). Since boundary state descriptions can be quite difficult
for orientifolds, any constraints are very useful, and as we show
in \cite{DMDR}, most of the sources can often be determined from the
$KR$-theory using what we know about $O^\pm$-planes.  

As noted in the introduction, one of our original motivations for a
detailed analysis of $T$-duality via orientifold plane charges in
$KR$-theory was the special case of $c = 3$ Gepner models as studied
in \cite{Bates:2006}. The authors of that paper used simple current
techniques in CFT to construct the charges and tensions of Calabi-Yau
orientifold planes. Using twisted $KR$-theory with a sign choice to
classify the brane charges does not depend on the specific structure
of $c = 3$ Gepner models, nor even on a rational conformal field
theoretic description. In \cite{2010arXiv1012.1634E} a twisted
equivariant $K$-theory description of the $D$-brane charge content for
WZW models is provided. Current work in progress attempts to
generalize this work by establishing an isomorphism between a suitable
(real) variant of twisted equivariant $K$-theory, sufficient to
capture orientifold charge content, and our $KR$-theory with sign
choices for Gepner models.  Such an isomorphism would allow the
computation of twisted $KR$-theory with a sign choice for complicated
Calabi-Yau manifolds through a simpler computation at the Gepner
point. 

$KR$-theory with a sign choice provides a universal $K$-theory for
classifying $D$-brane charges. In addition to being able to describe
new orientifold cases it reduces to all other known classifications on
smooth manifolds when using the correct involution. This unifies the
$K$-theoretic classification of $D$-brane charges by not requiring one
to change $K$-theories for different string theories. While its
definition was motivated by a problem in physics, the last point
exemplifies why $KR$-theory with a sign choice is also interesting
mathematically. 

 $KR$-theory with a sign choice provides a framework for studying the
underlying structure of $K$-theory. It was very surprising to see that
$KSC$-theory (the $KR$-theory of $S^{0,2}$) can be described as a
\textit{twisting} of the $KR$-theory of $S^{1,1}$. While we have
explicitly shown that twisted $KR$-theory with a sign choice satisfies
all possible $T$-duality relationships for spaces where the compact
dimensions are a circle or a $2$-torus, in this paper we did not look
at why there are isomorphisms between the twisted $KR$-theories of
$T$-dual theories. The purpose of this paper was simply to set up the
necessary topology to correctly classify brane charges. In \cite{DMDR}
we explore why $T$-duality gives isomorphisms of twisted
$KR$-theory with a sign choice. The extra data that
we needed to include is contained in the geometry of $T$-dual
theories. 
 
 We have already seen how considering the geometry is important. The
 complex structure constrains what involutions are possible on a
 $2$-torus. Since the physical theory depends on the involution, the
 geometry of the torus constrains the allowable string
 theories. Another well known example that played a role in our
 analysis is the $B$-field, which is determined by the K\"ahler
 modulus. We were also compelled to explore more exotic twists in
 order to account for the $T$-duality of the type I theory without
 vector structure. Without the physical motivation we might not have
 considered looking at such additional mathematical structures. We
 have shown how such twistings must behave via an Atiyah-Hirzebruch
 spectral sequence. In \cite{DMDR} we give a more geometric
 reason for why such twistings must be included. 
 
 By exploring the underlying topology and geometry we were able to
 gain physical information and new evidence for hypothesized
 dualities. Additionally, this work shows how we can go in the
 opposite direction and use the additional structure of physics to
 gain insight into the underlying geometry and topology. This gives us
 a greater understanding of the interplay between the three
 structures: topology, geometry, and physics.

\bibliographystyle{hplain}
\bibliography{T2}
\end{document}